\documentclass[11pt,a4paper]{article}
\usepackage[margin=1in]{geometry}
\usepackage[colorinlistoftodos]{todonotes}
\usepackage{graphicx}
\usepackage{rotating}
\usepackage[autostyle, english=american]{csquotes}
\MakeOuterQuote{"}
\usepackage{comment}
\newcommand{\indep}{\rotatebox[origin=c]{90}{$\models$}}
\usepackage[utf8]{inputenc}
\usepackage[english]{babel}
\usepackage{todonotes} 
\usepackage{algpseudocode}
\usepackage{algorithmicx}
\usepackage{algorithm}
\usepackage{breakcites}
\usepackage{microtype}
\floatname{algorithm}{Procedure}

\newcommand{\beq}{\begin{equation}}
\newcommand{\balign}{\begin{align}}
\newcommand{\eeq}{\end{equation}}
\newcommand{\ealign}{\end{align}}
\usepackage{url}
\usepackage{fixltx2e}
\usepackage{float}
\usepackage{hhline}
\usepackage{caption}
\usepackage{graphicx}
\usepackage{amsmath}

\usepackage{amsthm}
\theoremstyle{definition}
\newtheorem{definition}{Definition}[section] 
\theoremstyle{remark}

\theoremstyle{definition}
\newtheorem{example}{Example}[section]

\theoremstyle{definition}

\theoremstyle{definition}
\newtheorem{question}{Question}[section]

\theoremstyle{definition}
\newtheorem{idea}{Idea}[section]

\theoremstyle{definition}
\newtheorem{modification}{Modification}[section]

\theoremstyle{definition}
\newtheorem{assumption}{Assumption
}[section]

\newtheorem{theorem}{Theorem}[section]
\newtheorem{lemma}[theorem]{Lemma}
\newtheorem{proposition}[theorem]{Proposition}

\begin{document}

\title{Counterfactual-based Incrementality Measurement in a Digital Ad-Buying Platform}

\author{
  Prasad Chalasani\\
  MediaMath\\
  \texttt{\small pchalasani@mediamath.com}\\
\and
  Ari Buchalter \\
  MediaMath\\
  \texttt{\small abuchalter@mediamath.com}\\
\and
  Jaynth Thiagarajan \\
  MediaMath\\
  \texttt{\small jthiagarajan@mediamath.com}\\  
\and
  Ezra Winston\\
  Carnegie Mellon University\\
}

\date{\today}
\maketitle

\begin{abstract}
The problem of measuring the true incremental effectiveness of a digital advertising campaign is of increasing importance to marketers. 
With a large and increasing percentage of digital advertising delivered via Demand-Side-Platforms (DSPs) executing campaigns via Real-Time-Bidding (RTB) auctions and programmatic approaches, 
a measurement solution that satisfies both advertiser concerns and the constraints of a DSP is of particular interest. 

MediaMath (a DSP) has developed the first practical, statistically sound randomization-based methodology for causal ad effectiveness (or Ad Lift) measurement by a DSP (or similar digital advertising execution system that may not have full control over the advertising transaction mechanisms).
We describe our solution and establish its soundness within the causal framework of counterfactuals and potential outcomes, and present a Gibbs-sampling procedure for estimating confidence intervals around the estimated Ad Lift. 
We also address  practical complications (unique to the digital advertising setting) 
that stem from the fact that digital advertising is targeted and measured via identifiers (e.g., cookies, mobile advertising IDs) that may not be stable over time. One such complication is the repeated occurrence of identifiers, leading to interference among observations. 
Another is due to the possibility of multiple identifiers being associated with the same consumer, leading to "contamination" with some of their identifiers being assigned to the Treatment group and others to the Control group. 
Complications such as these have severely impaired previous efforts to derive accurate measurements of lift in practice.

In contrast to a few other papers on the subject, this paper has an expository aim as well, and 
provides a rigorous, self-contained, and readily-implementable treatment of all relevant concepts.
\end{abstract}

\section{Introduction}

As global digital advertising budgets continue to grow, overtaking TV as the dominant form of advertising \cite{emarketer-2016},
the problem of measuring the \textit{effectiveness} of ad campaigns is ever-important: 
advertisers want to ensure the budgets they spend on ad campaigns across search, display, mobile, video, social, email, and other digital channels are truly responsible for driving desired consumer behaviors. In particular, they want to understand not just what consumer behaviors and business outcomes occurred \textit{after} exposure to advertising, but also what occurred \textit{because} of exposure to advertising, i.e., the true causal impact.

At first glance it might seem that the ad-impact measurement problem is  straightforward, at least in theory,  
given the availability of vast amounts of data generated around digital advertising. In fact, one of the defining characteristics of digital advertising is precisely that it is "addressable", meaning that individual events such as the delivery of an ad to a consumer (e.g., within a website, mobile app, or other environment) and the subsequent online actions of that consumer (e.g., product purchases) can be tracked and measured. However as we describe in detail in this paper, even defining what we mean by "the effectiveness of an ad campaign" can be tricky. Once we have an acceptable \textit{definition}, the practical \textit{measurement} of ad effectiveness is far from straightforward, especially when the ad impact measurement must be done within the constraints and practical considerations of a Demand-Side-Platform (DSP) such as MediaMath, or a similar buy-side digital advertising execution system, on behalf of clients (advertisers) running campaigns on its platform.

This paper presents a detailed account of MediaMath's methodology for causal measurement of ad effectiveness 
using the  framework of \textit{counterfactuals} and \textit{potential outcomes}, which has emerged as a solid foundation on which to develop measures of causality. As we describe in the related work (Sec. \ref{sec-related}), we are certainly not the first to consider digital ad effectiveness from this viewpoint, so we mention here a few aspects that are unique to our paper. Unlike other papers in the advertising field, this work provides what we hope is a highly readable and \textit{self-contained} presentation of the causal framework as applied to digital advertising. Most other work in this area present mathematical notation with little build-up of background and expect the reader to either already be familiar with the causal analysis framework, or leave the reader to their own devices to consult the literature to understand it. 

This paper is also the first to present a detailed, comprehensive, and self-consistent ad-measurement solution that can be implemented by a DSP (henceforth, we shall use term “DSP” to generically refer not only to Demand-Side Platforms, but to any digital ad-buying system).
A number of practical constraints and considerations required us to develop an innovative methodology for ad-impact measurement at MediaMath. One constraint arises from the fact that many digital ads are transacted through a real-time auction, where DSPs compete for opportunities (“bid opportunities ”) to deliver ads to consumers in real time. 
In general, a DSP does not have full visibility into the ad auction after submitting a bid to an ad exchange or Supply-Side Platform (SSP) on behalf of an advertiser, nor full control over the auction mechanism and outcome. 
This precludes certain solutions that exist in the literature, such as Ghost ads \cite{Johnson_Ghost_2015}, which may be more relevant for so-called “walled garden” platforms such as Google, Facebook, and Amazon, who can see and control both the “buy” and “sell” sides of the auction process. 
Another constraint is that the simplest possible way of measuring ad-impact, namely, randomizing the ad opportunity to test/control \textit{after winning} the ad auction, is unacceptable to most advertisers because they have already paid for the winning bid, and hence would be wasting significant ad spend if the opportunity were assigned to the control group. 
To address this concern, our methodology instead randomizes \textit{before bid submission}, and this introduces the complication that not every consumer in the test group is exposed to the ad (since some bids would lose in the auction). 
This phenomenon is known as \textit{non-compliance} in the clinical trials literature (where certain people in the test group do not take the drug being studied). 
Non-compliance complicates causal ad-impact measurement because in general winning in the auction is not a random process, so consumers in the test group who are not exposed to ads may have systematically different response profiles than those who were exposed. This again owes to auction dynamics: unexposed consumers are those for whom the DSP was outbid in the auction, presumably because another buyer thought that consumer was more valuable (i.e., responsive). In other words, the exposed consumers in the test population are subject to a \textit{selection bias}, and in our context we refer to this as \textit{win bias}. 
We develop a methodology that builds upon some ideas from the clinical trials literature, adapting them to the advertising context in a manner that accounts for this bias.

In addition to merely producing a \textit{point estimate} of causal impact, it is also important to compute a suitable \textit{confidence interval} around that estimate. We adopt and simplify a Gibbs-sampling-based scheme from the causal analysis literature (specifically \cite{Chickering1996}) to compute this confidence interval. This paper presents a self-contained introduction to Gibbs-Sampling and is the first to describe a readily-implementable algorithm applying this technique to the problem of estimating confidence intervals on the causal effect. 

We also consider some critical real-world complications that arise in ad-effectiveness measurement that, to our knowledge, have not been adequately discussed elsewhere, and present modifications of our basic methodology to handle these complications. These have to do with the fact that digital advertising is delivered to consumers using devices and browsers that are tracked by identifiers such as cookies (within web browsers) and mobile advertising IDs (such as Apple IDFAs or Google Android IDs within mobile apps) which may or may not be stable over time. These identifiers constitute the data foundation for determining whether an ad was delivered to a particular device/consumer and whether a corresponding action (such as a purchase) was taken by that device/consumer. One complication we consider is the occurrence of the same identifier in different bid opportunities over time. This means that the responsiveness associated with the identifier in one bid opportunity may be affected by whether that identifier was exposed to an ad in a recent bid opportunity, violating a fundamental assumption behind the causal framework. Another complication is the presence of multiple identifiers for the same human user, causing two widespread forms of "ID contamination" -- "cookie contamination" and "cross-device contamination". Cookie contamination refers to the fact that the primary identifiers used to track browser-based usage \textit{on a single device} are typically not persistent over time, either because the browsers do not allow cookies to track them in the first place or because browsers or users periodically clear out the cookies. Either way, this means the same consumer will appear to have multiple cookie identifiers, some of which will inevitably and unknowingly be placed in Test and some in Control, leading to contamination of the populations. Cross-device contamination refers to the fact that even if all identifiers associated with a given device were stable, consumers don’t just encounter ads and make purchases on one device. The average US consumer owns nearly 4 connected devices (such as laptop and desktop computers, smart-phones, connected TVs, and gaming consoles; see   \cite{buckle-2016}),
and that is likely to grow over time with trends such as the Internet of Things, digital homes, wearable devices, etc. Absent certain knowledge of which devices are owned by which consumers, it is inevitable that some of the identifiers for a given consumer will correspond to device placed in Test, and some in Control, again leading to contamination. Moreover, these two forms of ID contamination are not mutually exclusive, which further compounds the problem.

We believe that the various issues and complications noted above have critically undermined most previous attempts to measure causal ad impact. Notably, MediaMath has observed many instances among its clients and their measurement partners, where measurement efforts in practice have yielded  no lift, or even negative lift! In fact, such results are the norm, and in the few cases where strong, positive lift is observed, it has tended to be short-lived, fluctuating strongly in time. These  results fly in the face of not only human intuition, but billions of dollars in ad spending. In developing the first practical and self-consistent methodology for DSP measurement of true causal ad effectiveness, we present solutions to all of the aforementioned issues and complications, and have in fact observed significant, positive, and stable lift when applying our methodology in practice, as will be demonstrated here.

It is worth noting here that while MediaMath currently engages predominantly in the execution of various forms of so-called “display” advertising (banners, videos, native ads, and other formats delivered on content websites, in mobile apps, on social platforms, and other digital media environments, across smart-phones, computers, and other connected devices), the methodology and techniques described in this paper can in principle be applied to all forms of digital advertising, including display, search, email, etc.

The following is an outline of the Paper. Section \ref{sec-define} starts with a detailed treatment of the definition of ad impact, and presents the causal framework of \textit{potential outcomes} and \textit{counterfactuals}. 
In Section \ref{sec-rtb} we describe the setting to which this paper applies, i.e. a DSP executing digital advertising, and also the types of data that need to be logged in order to measure ad impact.
Section \ref{sec-psa} starts with the simplest possible method to measure ad-impact, which is \textit{post-bid randomization}, and points out why this is a wasteful approach.
Then Section \ref{sec-prebid} describes MediaMath's  \textit{pre-bid randomization} approach and Subsection \ref{sec-causal-pre-bid} presents the problem of causal effect measurement in this scenario. Measuring ad impact under pre-bid randomization leads to the phenomenon of \textit{non-compliance} (which manifests here as win bias),
and Section \ref{sec-win-bias} introduces the mathematical machinery 
needed to conduct causal analysis under non-compliance. 
This section contains the core of our methodology for computing a point-estimate of Ad Lift. 
The Gibbs Sampling scheme for computing confidence-intervals is presented in detail in Section \ref{sec-gibbs}. 
Sections \ref{sec-repeat} and \ref{sec-contam} describe how our methodology needs to be modified to handle the complications of recurring identifiers and ID contamination, respectively. 
In Section \ref{sec-results} we present experimental results of actual Ad Impact measurements (and confidence intervals) for several campaigns, and show a detailed numerical example of our methodology applied to one campaign.
Section \ref{sec-related} discusses related work, and  
Section \ref{sec-future} concludes with an outline of future work.

\section{Defining Ad Impact} \label{sec-define}

Advertisers typically run “ad campaigns” to generate awareness and interest in their products, and influence consumers to buy them. One of the fundamental questions an advertiser wants to answer is: 
\begin{quote}
How effective is my advertising campaign?	
\end{quote}

It is worth noting that we are interested in measuring the impact of a \textit{specific ad campaign}. In general consumers are exposed to a variety of advertisements, both offline and online (which we can think of collectively as "background noise") and we are interested here in measuring the effectiveness of a specific campaign. However, we note that the methodology presented here can be adapted to measure the effectiveness of multiple campaigns, of a particular digital channel (e.g., search or display), or across all of an advertiser’s campaigns across all digital channels. These generalizations will be explored in later work. 

We will present a sequence of increasingly precise formulations of the above question, setting the stage for a rigorous statistical framework. 
Any quantification of ad effectiveness must specify a {\em desired outcome} that the advertiser wishes to elicit when exposing consumers to ads. In other words, advertisers want to quantify:

\begin{question}\small
What is the impact of {\em exposure to my ad campaign} in driving my {\em desired outcome}? 	
\end{question}

The desired outcome is a specific consumer behavior defined by the advertiser, such as a site visit, registration, subscription, addition of items to a shopping cart, purchase, etc. We use the general terms {\em response} or {\em conversion} to refer to such a desired behavior, and we use the term  {\em response rate} to generically denote the probability of a response, either at the individual (ad) level or aggregate (campaign) level (these will be defined more rigorously later).

It is important to emphasize causation here: for instance if the advertiser finds that consumers exposed to their ads, on average, have a response rate of 3\%, then it does not necessarily follow that this 3\% response rate (sometimes referred to as the "aggregate" or "overall" or "top-line" response rate) was entirely {\em caused} by exposure to their ad. Some of those conversions might have occurred anyway in the absence of exposure to the ad campaign being measured; consumers might have visited the website, purchased the product, etc. without having seen any ad at all, or after seeing an ad from a different campaign than the one being analyzed. So a more relevant question from the advertiser's perspective is, 

\begin{question}\small
	How much of the response rate of exposed consumers is  \textit{caused} by my ad campaign? 
\end{question}

For instance it is possible that even without seeing the ad, these {\em exposed} consumers {\em would have had} a response rate of 2\%. Thus the {\bf incremental} effect of this campaign is only 1\%, and it would be reasonable to say that out of the 3\% total response rate of consumers exposed to the campaign, only 1\% is {\emph caused} by this campaign.\footnote{This point is especially germane in the context of display advertising, also sometimes referred to as "banner advertising". For display ads, a consumer might click on an ad, thereby leading them to perform some downstream conversion behavior. In this case, the presence of a click usually implies some casual relationship between the ad and the conversion. However, the link between the display ad and the ultimate conversion may be less direct; consumers seeing these ads may instead perform a search related to the ad, directly navigate to a website or mobile app, or simply be more inclined to make an online or offline purchase in the future as a result of greater awareness and/or affinity.}

It is precisely this incremental (or causal) effect that advertisers seek to measure; they want to direct ad spending towards  campaigns shown to have larger incremental effects, or target consumers whose incremental response is likely to be higher.

Thus a more precise formulation of the ad effectiveness question is:
\begin{question} \label{q-incr}\small
	How much higher is the response rate $R_e$ of exposed consumers, compared to $R'_e$, the response rate {\em they would have had} if they had {\em not} been exposed to the ad campaign? In other words, what is the \textbf{incremental} effect attributable to the ad exposure?
\end{question}

Clearly it is not possible to observe $R'_e$, the response rate exposed consumers would have had if they had not seen the ad (in exactly the same context, i.e. time, location, website, etc.), and so this is called a {\em counterfactual} response rate. 
Nevertheless, it is possible to measure the causal effect under some conditions, which we highlight here, and make more precise later:
\begin{idea}\label{idea-stat}\small
If the exposed and unexposed populations are {\em statistically equivalent} then we can validly compare their response rates to measure the causal effect of the campaign. 
\end{idea}
Note that if we simply take the difference in response rates of the exposed and unexposed populations during the normal course of running an ad campaign, this would in general \textit{not} yield a valid measure of incrementality since the two populations would not be statistically equivalent (e.g., most commonly due to the presence of campaign targeting settings that apply to the entire exposed population, but not to the entire unexposed population).
There are, however, approaches to infer causality or incrementality from such \textit{observational data} (see, e.g. \cite{Austin_An_2011} and other references in the Related Work Section \ref{sec-related}). These approaches tend to be highly assumption-driven, prone to biases, and difficult to validate in practice, and are therefore best suited for incrementality measurement at the individual treatment level (i.e., consumer level), where counterfactuals cannot be established. At the aggregate level (where statistically equivalent counterfactual samples for Test and Control can be defined), randomized experimental testing is preferred over observational methods, and the aim of this work is to derive unbiased estimates of causality using an experimental approach that builds on established frameworks, accounts for real-world factors, and can be easily reproduced and verified.

\subsection{Potential Outcomes and Counterfactuals} \label{sec-potential}

Estimating causal effects is a fundamental problem in many fields, and the framework of counterfactuals and potential outcomes has emerged as 
the most widely used one for causal analysis; e.g., see \cite{Little_Causal_2000} for a thorough introduction. 
We adopt notation from this literature in what follows.

For any individual consumer $i$, let $W_i$ be a binary random variable indicating whether the consumer sees the ad ($W_i = 1$) or not ($W_i = 0$). Borrowing from the clinical trials literature, when $W_i=1$ we say the consumer is {\em treated} (i.e., in this case, exposed), otherwise the consumer is {\em not treated} (unexposed).
\footnote{Although in the clinical trials literature the terms treated/untreated are common, we will instead use the terms exposed/unexposed in most contexts.} 
We also define a binary random variable $Y_i$ to denote the consumer's {\em response} to the treatment $W_i$: If the consumer {\em responds} to the treatment (i.e. performs the desired behavior, such as a conversion, etc.), $Y_i(W_i)=1$, and otherwise, $Y_i(W_i)=0$. Note that $Y_i(1)$ represents the response of consumer $i$ under treatment, and $Y_i(0)$ is the response under no-treatment. These two possible values are called {\bf potential outcomes}. Figure \ref{fig-counterfactuals} 
illustrates these concepts, and others introduced in this section.

\begin{figure}\centering
\includegraphics[width = 0.8\textwidth, angle = 270]{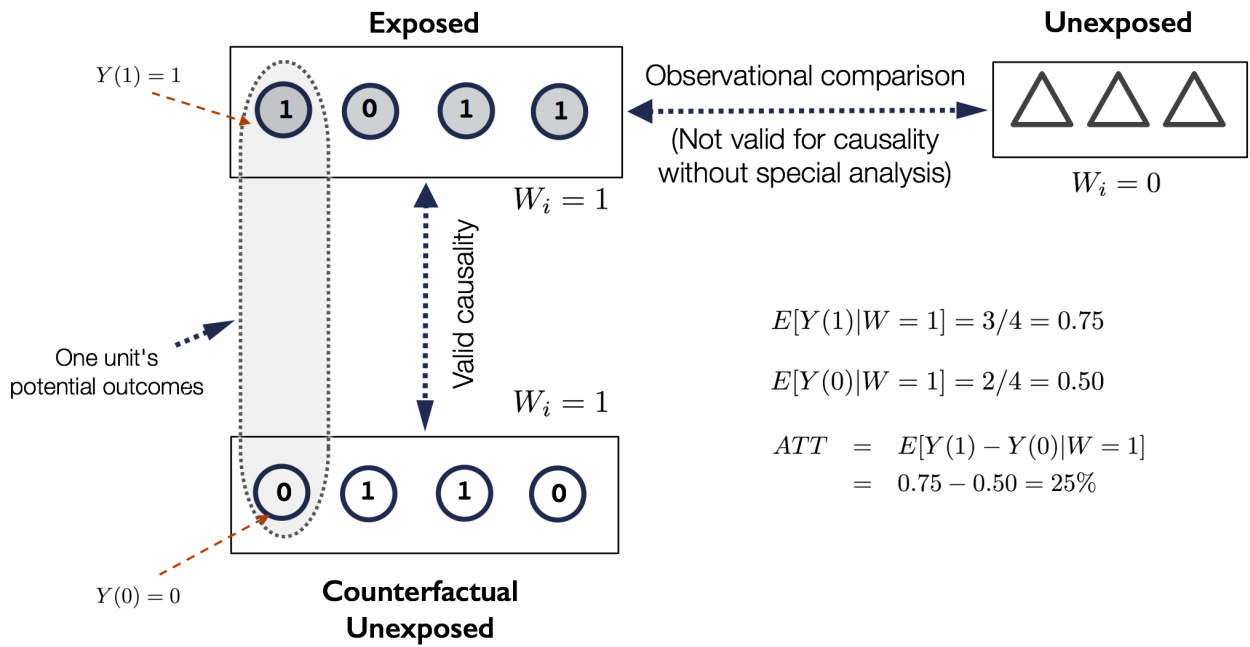}
\caption{\small The definition of causal ad impact  in terms of counterfactuals and potential outcomes. 
Each circle or triangle represents a "unit" in the analysis (taken to be a consumer represented by an identifier or "userID"). 
Circles denote the units $i$ in the population that happen to have been exposed during the normal course of running an ad campaign -- these have  $W_i=1$. Triangles denote the units $i$ in the unexposed population, for which $W_i=0$. 
The different shapes emphasize that these populations are not necessarily statistically equivalent. 
The numbers (0 or 1) inside a shape shows the value of the potential outcome of a unit. 
The rectangle denoted "exposed" contains the if-exposed potential outcomes of the exposed population, i.e. the $Y_i(1)$ values, which are observable.
The rectangle denoted "counterfactual unexposed" contains the potential outcomes of the exposed population \textit{had they not been exposed} to the ad campaign, i.e.,  the $Y_i(0)$ values, which are unobservable counterfactuals. 
The difference in response rates of the exposed population ($E[ Y(1) | W=1]$) and that of the Counterfactual Unexposed population ($E[ Y(0) | W=1]$) is the Average Treatment Effect on the Treated ($ATT$), which is a valid definition of causal effect. A simple numerical example of this calculation is shown.
}
\label{fig-counterfactuals}
\end{figure}

Note that a specific consumer $i$ is either treated or not, so precisely one of the two values $Y_i(0)$ or $Y_i(1)$ is observed. This is frequently formalized as:

\begin{quote}
	{\bf Fundamental Problem of Causal Inference:} For any individual $i$, we can never directly observe both potential outcomes $Y_i(0)$ and $Y_i(1)$. 
\end{quote}

   In particular the outcome $Y_i(W_i)$ is observable, whereas the outcome $Y_i(1 - W_i)$ is a {\bf counterfactual}, unobservable outcome. Although we consider more nuanced situations later, for the purpose of this initial mathematical formulation we ignore issues around the timing of the treatment and response of different consumers. For now we can imagine an idealized scenario where for all consumers $i = 1,2,\ldots,n$, the $W_i, Y_i$ are observed simultaneously at a specific instant in time. 

\subsection{Causal Effects}

We can now define various quantities:

\begin{definition}\label{ice}\small
{\bf Individual Causal Effect} (ICE) for a consumer $i$:
$$
ICE_i = Y_i(1) - Y_i(0).
$$
\end{definition}

From the Fundamental Problem, it follows that it is impossible to directly compute the ICE (although as mentioned in the Related Work Section \ref{sec-related}, there are recent machine learning based approaches to estimate the ICE as a function of features, or "covariates", of an individual unit). A more modest goal would be to estimate the \textit{average} causal effect over a group of consumers. In order to properly define an average, we need to assume some distribution over userIDs $i$, and we will assume the simplest one:

\begin{assumption}\small \textbf{Uniform Distribution
\footnote{The assumption of a uniform distribution is not restrictive at all: when there is a group of individuals being studied, and we want to express averages over this group as expectations, the uniform distribution is a straightforward device that provides the simplest way of doing this. The intention is not to specify that individuals occur "in reality" with this distribution.}
over UserIDs}
\label{ass-unif}
In all probability and expectation computations, we assume a uniform distribution over userIDs $i$, i.e. all of them occur with equal probability.
\end{assumption}

Now we can define what we mean by the average causal effect:

\begin{definition}\label{ace}\small
{\bf Average Causal Effect} (ACE):

$$
ACE = E_i[ ICE_i] = E_i [Y_i(1) - Y_i(0) ],
$$
\end{definition}
where the expectation $E_i[.]$ is taken over the distribution of consumers $i$. In other words, the ACE is the population-level average response rate if all consumers had been exposed to the ad, minus the response rate if none were exposed. We sometimes suppress the consumer subscript $i$ and simply write 
$$
ACE = E[ Y(1) - Y(0) ] = E [Y(1)] \; - \; E [Y(0)]
$$

From the previous discussion it will be evident that the $ACE$ is not exactly what we are after: it represents the incrementality across {\em all} consumers, exposed and unexposed. However in Question \ref{q-incr} we are interested in the incrementality of {\em exposed} consumers only. We can define this by simply conditioning the $ACE$ on $W=1$, and this leads to:

\begin{definition}\label{att}\small
{\bf Average Treatment Effect on the Treated} (ATT):

\begin{eqnarray}
ATT & = & E_i[ ICE_i | W_i = 1] \\
    & = & E_i [Y_i(1) - Y_i(0) | W_i = 1] 	\\
	& = & E[ Y(1) - Y(0) | W=1] \\
	& = & E[ Y(1) | W=1] - E[ Y(0) | W=1]
\end{eqnarray}

\end{definition}

Under what conditions can we estimate the $ATT$? In the last expression in the definition above, the first term $E[ Y(1) | W=1]$ is the observable average response rate of exposed consumers, while the second one $E[ Y(0) | W=1]$ is a counterfactual: we cannot directly observe the unexposed potential response rate of exposed consumers. What if we try to use the observed response rate of {\em unexposed} consumers $E[Y(0) | W=0]$ in place of the counterfactual second term? In other words we can try to estimate the $ATT$ by using the so-called {\em standard estimator:} 
$$
S = E[ Y(1) | W = 1]  \; -  \; E[ Y(0) | W = 0 ],
$$
which is the difference between the average response rates of treated and untreated consumers, both of which are observable. In general, however, $S$ would not be equal to the ATT: The ATT measures {\em causation} whereas S merely measures {\em association} \cite{elwert2013graphical}. To see why $S$ may not always estimate $ATT$ correctly, consider the following decomposition of $S$:

\begin{eqnarray*}
S &= & E[ Y(1) | W = 1]  \; -  \; E[ Y(0) | W = 0 ] \\
  &= & \underbrace{E[ Y(1) | W = 1]  \; -  \; E[ Y(0) | W = 1 ]}_{ATT} \; \; + \\ 
  &  & \underbrace{E[ Y(0) | W = 1]  \; -  \; E[ Y(0) | W = 0 ]}_{\text{selection bias}}  
\end{eqnarray*}

In other words, the difference in observed response rates between exposed and unexposed consumers is the sum of the $ATT$ and the selection bias. 
When there is a large selection bias, $S$ is not a good estimate of $ATT$ because the average observed response rate of unexposed consumers is a poor substitute for the counterfactual $E[ Y(0)|W=1 ]$. An example is shown in Fig. \ref{fig-att-bias}.

\begin{figure}\centering
\includegraphics[width=0.9\textwidth]{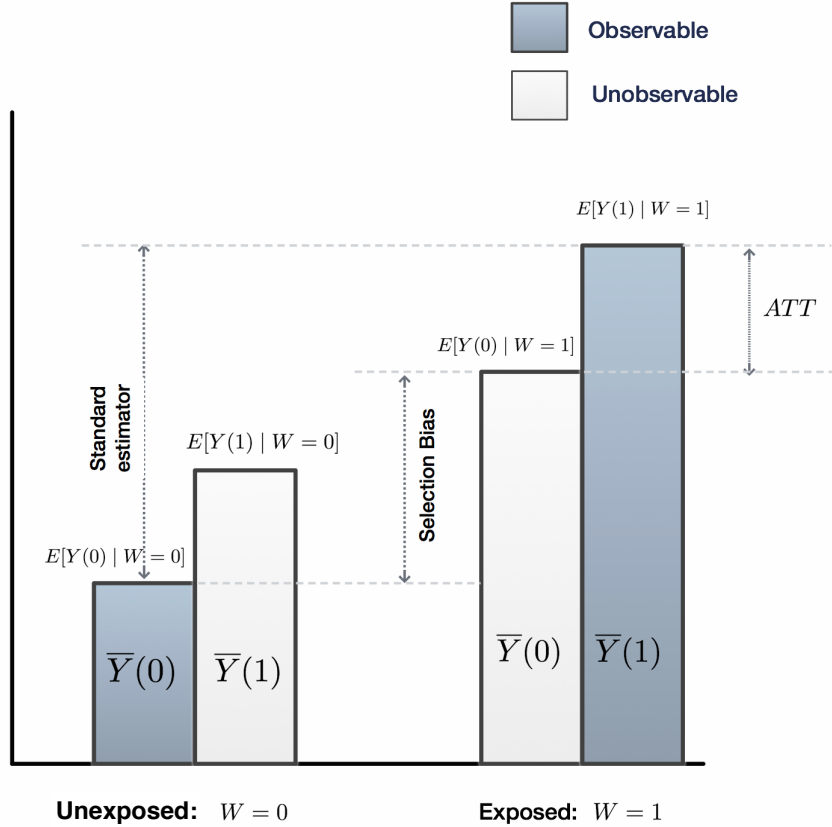}
\caption{\small Selection bias and $ATT$. The height of each bar represents the expected value of a potential outcome $Y(0)$ or $Y(1)$ for a population.
 The left pair of bars correspond to the unexposed population, $W=0$, and the right pair of bars correspond to the exposed population, $W=1$. 
 The Standard Estimator is a measure one would naively use to attempt to measure causal effect, i.e. the difference between the two observable average outcomes. 
 This figure shows pictorially how this naive measure can be decomposed into a combination of Selection Bias and the true causal effect, or $ATT$.
 A non-zero Selection Bias indicates that the baseline ("if unexposed") response rates of the unexposed population and exposed population are different, as shown by the difference in height of the left bars in each pair.}\label{fig-att-bias}
\end{figure}

We saw above that when the selection bias is zero, then the standard estimator $S$ equals the $ATT$. Absence of selection bias is a special case of a more general condition under which we can estimate $ATT$ (which involves a counterfactual) from observable expectations:

\begin{definition}\small [Ignorability of Exposure] \label{def-ignor}
	When $(Y(0), Y(1) ) \indep W$, i.e. when the potential outcomes $Y(0), Y(1)$ are jointly independent of exposure $W$, we say that exposure is {\bf ignorable} with respect to potential outcomes. For brevity we will simply say "exposure is ignorable" but it should be understood that we mean ignorability with respect to potential outcomes.
\end{definition}

It is important to realize that ignorability of exposure does not imply that the \textit{actual} outcome $Y_i$ of an individual $i$ is independent of exposure $W_i$; 
in general the actual outcome \textit{will} depend on exposure-status. Ignorability of exposure is a statement about \textit{potential} outcomes: 
the \textit{distribution} of potential outcomes is independent of exposure, i.e. the overall distribution of potential outcomes $(Y(0), Y(1))$ across the population is identical to their distribution in the exposed sub-population (where $W=1$) and in the unexposed sub-population (where $W=0$). The following property directly follows from this:

\begin{proposition}[Ignorability and Expectations]\label{prop-ignor}\small
When ignorability holds, the expectations of $Y(0)$ and $Y(1)$ are not affected by conditioning on exposure:
\begin{eqnarray} 
  E[Y(1) | W=1] = E[ Y(1) | W = 0] = E[Y(1)],\\
  E[Y(0) | W=1] = E[ Y(0) | W = 0] = E[Y(0)].  \label{eq-ignor}
\end{eqnarray}

\end{proposition}

The ignorability property is sometimes referred to as {\em exchangeability} to highlight the fact that the treatment status $W=1$ and $W=0$ can be interchanged when taking conditional expectations with respect to values of $W$. 
Thus ignorability makes precise our informal notion of {\em statistically equivalent populations} that we alluded to in Idea \ref{idea-stat} earlier. Ignorability has an important consequence for our purposes:

\begin{lemma} \label{lem-ignor-att}\small

\small When ignorability of exposure holds, the $ATT$ can be estimated from observable averages as 
$$
E[Y(1) | W = 1] - E[ Y(0) | W = 0]
$$
\begin{proof}\small
Recall that the $ATT = E[Y(1) | W=1] - E[Y(0) | W=1]$, where the second terms is a non-observable (counterfactual) expectation, and Eq. \ref{eq-ignor} implies it can be replaced by $E[Y(0) | W=0]$, and the result follows.
\end{proof}
\end{lemma}

\subsection{Ensuring Ignorability by Randomization}

We saw in the previous section that ignorability of exposure ensures that the $ATT$ can be measured from the observable average response rates of Test and Control groups. The easiest way to ensure ignorability is by randomization: when the treatment variable $W_i$ is assigned 0 or 1 randomly (not necessarily with equal probability), then clearly the potential outcomes $Y_i(0), Y_i(1)$ are independent of $W_i$. This fact motivates the first simple setup for measuring the $ATT$ for an ad campaign, which we will describe in Section \ref{sec-psa}, but first we will pause briefly and outline a basic picture of how a  digital ad-buying platform (i.e., a DSP) operates.

\section{Demand Side Platforms} \label{sec-rtb}

We provide here a simplified view of the mechanisms involved in the operation of a DSP, as pertains to buying and delivering digital advertising. We will specifically describe the buying mechanism of Real Time Bidding (RTB), but the general principles laid out here apply to any ad-delivery contexts where similar event-level logging (as described below) and data flows exist, and the resulting methodology can be readily applied to those contexts as well. 
Most free (and some paid) websites and mobile apps depend on advertising for revenue. 
When a consumer visits a publisher's website or mobile app, this becomes an {\emph opportunity} to display an ad to the consumer. The publisher sends this opportunity to one or more ad exchanges or SSPs (which we will simply refer to collectively as "exchanges").
The ad exchange then sends these opportunities, also known as bid opportunities or "bid requests", along with associated data about the request, to various DSPs, whose clients are advertisers (or their agencies) that compete to win the opportunity. 
The bid request can be viewed as a tuple containing several fields, and in our present context the only ones that matter are $(u, r, e)$ where:
\begin{itemize}
	\item $u$ is the {\bf userID}. For example, this could be a cookie ID from a web browser, or a mobile advertising ID (or "deviceID" in brief) from a mobile app. 
	\item $r$ is the {\bf bid request ID}, which uniquely identifies the bid request, and we can imagine that all information needed to serve the ad will be attached to the bid request ID.
	\item $e$ is the {\bf exchange ID} from which the bid request originated.
\end{itemize}

A DSP typically has hundreds or thousands of clients (advertisers) on whose behalf it submits bids throughout the day. 
However in response to a given bid request, only a subset of advertisers may be eligible to bid due to several factors such as: (a) advertisers' targeting requirements, i.e., the types of consumers they want to reach, the contexts in which they want to reach them, etc. (b) governing campaign criteria such as the available budget and desired frequency of exposure, or (c) publisher restrictions on which types of advertisers they will accept. 
The DSP therefore needs to conduct what we will refer to generically as a {\em matching/targeting} process to arrive at a short list of advertiser campaigns eligible to submit a bid. 

If the DSP determines there are one or more eligible advertiser campaigns, it must then determine the appropriate price to bid for the impression in question on behalf of the eligible campaign(s). Determination of the optimal bid price to submit is an interesting problem in its own right, but we ignore the details of the bid determination process here. 
For our present purpose the only relevant aspect of the result of the optimal bid price determination process is that it is either empty (for example if no ad campaigns are eligible) or it is a {\bf bid submission} tuple 
$(u,r,e,c,b)$ which contains the entries $u,r,e$ from the bid request, plus two additional entries $c,b$ where:
\begin{itemize}
	\item $c$ is the {\bf campaign ID} on behalf of which the bid is submitted. Note that the campaign ID implies a unique advertiser ID; each advertiser typically runs a collection of campaigns, each with its own unique ID. 
	\item $b$ is the {\bf bid amount}, denominated in some monetary currency.
\end{itemize}

We will treat the combination of matching/targeting and bid-determination as a black box called the {\em Bidder}, whose input is the bid request tuple 
$(u,r,e)$ and the output is either empty or a bid-submission tuple  
$(u,r,e,c,b)$ that is sent to the exchange $e$ from which the bid request was received. 

There would in general of course be several other entities (e.g., multiple DSPs) bidding for the same bid request $r$ from the exchange $e$, and the exchange conducts an auction to determine the winner. The auction process is a black box for our purposes, and the end result is that the DSP receives an {\bf auction-result} tuple $(u,r,e,c,a,x)$ whose entries $u,r,e,c$ are from the bid-submission tuple, and 
\begin{itemize}
	\item $x$ is the {\bf auction outcome}, a binary variable indicating whether the bid submission won the auction ($x= \text{Win})$ or not ($x=\text{Lose}$).
	\item $a$ is the {\bf clearing price}, i.e. actual amount the DSP must pay the exchange (this is typically lower than the bid amount $b$ as most RTB auctions are some form of modified second-price auction). The specifics of the relation between $a$ and $b$ (i.e., the auction clearing mechanism) are not relevant to the measurement of causal ad-effectiveness which concerns us here.
\end{itemize}

Finally, if the DSP won the auction, it will serve an ad from campaign $c$ to the consumer represented by identifier $u$ using the session information associated with the bid request $r$, and log a record of the win. If it lost the auction, there is no action taken, other than possibly logging a record of the loss.

\subsection{Logging}

In order to be able to optimize the operation of all aspects of a DSP (including determination of optimal bid prices, implementing targeting, etc.), it is essential that the system conducts {\em logging} of not only of bidding events but also of consumer behaviors, via suitably placed tracking mechanisms (also known as "beacons" or "pixels") on advertiser websites, mobile apps, and other digital properties. Machine learning algorithms can make use of such logs to improve the efficacy of the DSP algorithms over time. It turns out that logging is also crucial for measuring the causal effect of ads. 

We assume the existence of 3 logs:
\begin{itemize}
	\item {\bf Bid Opportunity Logs}: When the DSP receives a bid opportunity, a log is created containing the tuple $BidOpp(u, t, r, e)$, where $u$ is the userID, $t$ is the time-stamp, $r$ is the bid request ID, and $e$ is the exchange ID. In general of course the number of bid opportunities seen daily can easily run into the hundreds of billions, and could be prohibitively large to store fully, so we assume here only the ability to ingest the stream of bid opportunities and store what we need.
	\item {\bf Impression logs}: When the DSP serves an ad for a campaign $c$ at time $t$ in response to a bid request (if it won the auction) for userID $u$, a log is created containing the tuple $Imp(u,t,c)$. Typically impression logs can contain a variety of other variables (especially relevant to ML algorithms) but those are ignored in the present discussion.
	\item {\bf Event logs}: When a userID $u$ performs a conversion action (e.g., a purchase on a suitably instrumented mobile app or web-page) relevant to a specific campaign $c$ at time $t$, a log is created containing the tuple $Ev(u, t, c)$. Essentially, we are logging the fact that userID $u$ had a "response" relevant to campaign $c$, and crucially, $u$ may not actually have been exposed to an ad from campaign $c$. Here we only consider actions relevant to defined campaign goals, and ignore other events.
\end{itemize}

These logs can be augmented to help with causal effect estimation, as we will see later.

\subsection{Campaign Post-View Windows}\label{sec-pv}

One aspect of consumer "responses" we have glossed over, is the timing of response. Returning to the potential outcomes framework presented earlier, suppose we start observing a userID $u_i$ at time $t = t_0$. At this time the consumer is either exposed to an ad (from the campaign $c$ under consideration) or not, i.e. $W_i = 1$ or $W_i=0$. How do we define the response variable $Y_i(W_i)$? To properly define the response variable we introduce the notion of a "window of influence" $V_c$ of a campaign $c$. In theory, such a window isn't strictly necessary, as the response of consumers can simply be expected to diminish in some rapid but smooth (e.g., perhaps exponential) function of time. In practice however, most advertisers impose a hard cutoff such that they are willing to ascribe impact to an ad exposure for times $t < V_c$ but unwilling to do so for times $t > V_c$. The specific cutoff value of $V_c$ varies by advertiser, but is assumed to be sufficiently long as to capture substantially all of the impact. Thus, the advertiser is interested in measuring the causal effect of their ad up to $V_c$ time units post-exposure. We call $V_c$ the {\bf post-view window} (PV window for short) of campaign $c$. Roughly speaking, if we log an impression $I = Imp(u,t,c)$, and subsequently log an event $E = Ev(u,t',c)$, then we can validly attribute the event $E$ to impression $I$ if $t' - t \leq V_c$. We will discuss more carefully later how to use the PV window in various scenarios, in the context of measuring response rates of exposed and unexposed consumers.

\section{The Simplest Randomization: Post-Bid} \label{sec-psa}

With the above description of how a DSP operates, we are ready to specify the simplest possible design of a randomized test system to determine the causal effect of an ad campaign $c$, i.e., the $ATT$. 
The basic idea is simple (see Figure \ref{fig-post-bid}):

\begin{figure}\centering
\includegraphics[width = 0.6\textwidth]{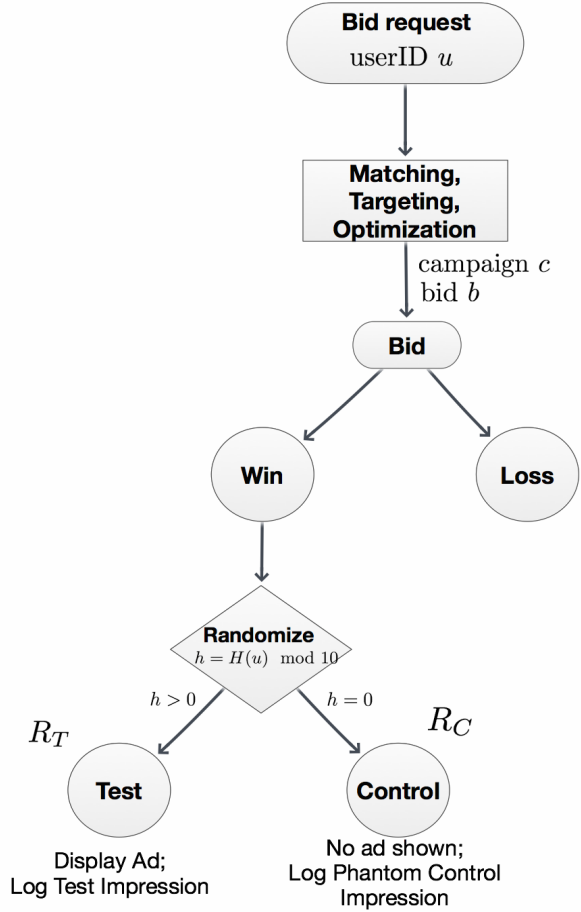}
\caption{\small Populations and response rates involved in the post-bid randomization scheme.}
\label{fig-post-bid}
\end{figure}

\begin{definition}\small
{\bf Post-Bid Randomization:} Fix a {\em holdout fraction} $p < 1$. 
A "random" (in the sense below) fraction $p$ of userIDs who are {\em about to see an ad} for campaign $c$ (i.e., where the DSP has won the auction), are instead shown a PSA (Public Service Announcement) ad\footnote{Alternatively, these consumers could be shown no ad at all, which would leave the methodology and results here unchanged, but in practice once a buyer wins an auction they are compelled by publisher to produce some ad, lest the page render a blank box which would constitute a poor user experience.} that is unrelated to campaign $c$ (and assumed to have no impact on the campaign's desired outcome). The $(1-p)$ fraction who are exposed to the campaign are called the Test group, and the fraction $p$ who receive PSA ads are called the Control group.
\end{definition}

It is important to realize that the Test/Control assignment is assumed for now to be a function of the userID only, i.e., a userID is either a Test or a Control userID. So if the same userID appears multiple times in different bid requests, it will either always be in the Test group or always be in the Control group. 

The Test/Control assignment can be accomplished by picking a suitable {\em deterministic} uniform hash function $h_k(u)$ that maps a userID $u$ to a $k$-digit decimal number, and determine Test/Control assignment based on the last few digits. The hash function is uniform in the sense that if we pick a random userID $u$, then $h_k(u)$ is uniformly distributed over the range $0$ to $10^k-1$. For example, if the holdout fraction $p = 10\%$,  then we can map any userID $u$ that has $h_k(u)$ ending in 0 to Control, and otherwise to Test. Since the hash function values are uniformly distributed, and the last digit of $h_k(u)$ is equally likely to be any of the 10 decimal digits, we will have mapped roughly 10\% of userIDs to Control. Note that we are not actually doing any explicit randomization. Instead, we are only assuming that the userIDs are independent of their potential outcomes (and hence the Individual Causal Effect, ICE), which is a very reasonable assumption. Since Test/Control assignment depends only on the userID via the hash function, it then follows that the Test/Control assignment is ignorable when estimating causal effects. In this sense, we can think of Test/Control assignment as being essentially random over userIDs.  

For future reference, we summarize this as pseudo-code in Procedure \ref{alg-post-bid} below. Here $c_0$ denotes the specific campaign for which we intend to estimate the causal effect $ATT$. To simplify the presentation we will assume existence of a Test/Control assignment function 
$H(u,p) \in \{C,T\}$ (where $C$ represents Control, $T$ represents Test)
such that if the userID $u$ is chosen at random, then $P_u[ H(u,p) = C]= p$. Such a function can easily be implemented by means of a uniform hash function $h_k(.)$ as described above.

\begin{algorithm} 
 \caption{\small Post-Bid Randomization to estimate causal effect}
    \label{alg-post-bid}
\begin{algorithmic}[1]
\Require 
$c_0$: CampaignID being measured, 
$p$: Holdout (control) Fraction,
$H(u,p)$ = test/control assignment function

\While {True} \Comment{{\em endless loop of bidder}}
\State {$(u,r,e) \gets $ Bid Request from exchange $e$;}
\State {$(u,r,e,c,b) \gets $ (Bidder match/targeting/bid-Optimization);}
\State {submit bid $(u,r,e,c,b)$ to exchange $e$;}
\State {$(u,r,e,c,a,x) \gets $ AuctionResult from exchange $e$;}
\State{$t \gets$ current time;}
\State{$g \gets H(u,p)$;} \Comment{$g \in (T,C)$}
\If {$x == \text{Lose}$} 
	\State{do nothing}
\ElsIf{$c == c_0$ and  $g == C$} 
	\State{serve PSA ad to bid request $r$;}
    \State{log a "phantom control impression" $Imp(t, u, c_0, C)$} \label{phantom}
\Else    
	\State{serve campaign $c$ ad to bid request $r$;}
	\If{group $ == T$} 
		\State{log a test impression $Imp(t, u, c_0, T)$}
	\EndIf
\EndIf
\EndWhile
\end{algorithmic}	
\end{algorithm}

Note that when a "phantom control impression" $Imp(t, u, c_0, C)$ is logged (line \ref{phantom}) we are recording the fact that at instant $t$, the userID $u$ was {\em about to see an ad} from campaign $c_0$, but in the "last millisecond" we decided instead to show the userID a PSA ad instead, because $H(u,p) = C$. Since the Test/Control assignment depends only on the userID $u$, and the potential outcomes of $u$ do not depend on $u$, it follows that Test/Control assignment is ignorable when computing the conditional expectations $E[Y(1) | W=1]$ and $E[Y(0) | W=1]$. Lemma \ref{lem-ignor-att} then implies that we can estimate the $ATT$, and we describe this next.

\subsection{Causal Effect Measurement}\label{sub-causal-post-bid}

For simplicity we assume the following for now:
\begin{assumption}\small \label{assum-user-once}
	No userID appears more than once in any of the impression logs during this period. In Sec. \ref{sec-repeat} we will describe a methodology to handle recurring userIDs.
\end{assumption}

In addition, we continue to assume a uniform distribution over userIDs (Assumption \ref{ass-unif}).

To estimate the causal effect ($ATT$) we need to estimate the response rates of the Test group ($W=1$) and Control group ($W=0$), i.e., 
$R_T = E[Y(1) | W=1]$ and $R_C = E[Y(0)|W=0]$. 
See Figure \ref{fig-post-bid} which schematically shows the main populations and response rates involved in the $ATT$ computation.
Suppose we wish to estimate these after having run a post-bid randomized test over some period. From the impression logs we know the exact time when each Test userID was exposed to an ad from campaign $c$, and 
when each Control userID was exposed to a PSA ad. Now from the event logs, we can tell which of these exposures lead to a response (i.e. $Y=1$), where we consider an event $Ev(u,c,t')$ to be a response to a Test impression $Imp(u,c,t,T)$ if $t'-t \leq V_c$, where $V_c$ is the PV window of the campaign $c$. An analogous criterion applies to Control impressions. Since we have assumed that each userID appears exactly once, we do not need to be concerned with a userID having multiple ad exposures, with possible overlap of PV windows (which would complicate assignment of "credit" for responses). This results in 4 numbers: 
\begin{itemize}
	\item $n_t, n_c$ the number of Test and Control impressions respectively
	\item $k_t, k_c$ the number of Test and Control responses respectively,
\end{itemize}

which gives estimates of the Test and Control group response rates as $R_T = E[Y(1)\,|\,W=1] = k_t/n_t$ and $R_C = E[Y(0)\,|\,W=0] = k_c/n_c$ respectively, and so we estimate $ATT = R_T - R_C =  k_t/n_t - k_c/n_c$. 
(see Figure \ref{fig-post-bid}).

\subsection{Wasted Spend in Post-Bid Randomization}
While the post-bid randomization idea is simple to implement and leads to a statistically sound estimation of the causal effect, it has one big drawback: 
since the Test/Control assignment is done {\em after} placing a bid, the advertiser must pay for the won impression, regardless of whether the consumer is served an actual campaign ad, or a PSA ad. 
Thus if the holdout fraction is 10\%, the advertiser essentially wastes 10\% of their ad budget on PSA ads, which in turn hurts campaign ROI by a corresponding amount. 
Moreover, given the low absolute value of typical response rates, larger holdout fractions may sometimes be required to improve measurement significance. 
Considering that annual digital advertising budgets for large advertisers can easily run into the 10s or even 100s of millions of dollars, this is a significant problem. 
Motivated by this consideration, we have designed a {\em pre-bid randomization} scheme, and a methodology to estimate $ATT$ under that scheme, which does away with the need to spend money on bids won against the control group. We describe this in detail in the following sections.

\section{Pre-Bid Randomization} \label{sec-prebid}

The idea here is to assign a userID to Test or Control {\em before} submitting a bid to the exchange, and only submit bids for userIDs assigned to Test. We make this precise in Procedure \ref{alg:pre-bid}, also shown schematically in Fig. \ref{fig-prebid-flow}

\begin{algorithm} 
 \caption{\small Pre-Bid Randomization to Estimate Causal Effect}
    \label{alg:pre-bid}
\begin{algorithmic}[1]
\Require 
$c_0$: CampaignID being measured, 
$p$: Holdout (control) Fraction,
$H(u,p)$ = Test/Control assignment function
\While {True} \Comment{{\em endless loop of bidder}}
\State {$(u,r,e) \gets $ bid request from exchange $e$;}
\State {$(u,r,e,c,b) \gets $ (Bidder match/targeting/bid-optimization);}
\State{$t \gets$ current time;}
\State{$g \gets H(u,p)$;} \Comment{$g \in (T,C)$}

\If{$c == c_0$ and $g == C$}
	\State{do not submit bid;}
	\State{log a "phantom control impression" $Imp(t,u,c_0, C)$}
\Else \Comment{{\em non-tested campaign, or Test userID for $c_0$ }}	 
	\State{submit bid $(u,r,e,c,b)$ to exchange $e$}
	\State {$(u,r,e,c,a,x) \gets $ AuctionResult from exchange $e$}	
	\State{$t \gets$ current time;}
	\If {$x == \text{Win}$} 
		\State{serve campaign $c$ ad to bid request $r$;}	
		\If {$c == c_0$}
			\State{log a Test WIN impression $Imp(t,u,c_0,TW)$;}
		\EndIf
	\ElsIf {$c == c_0$}
		\State{log a "phantom" Test LOSS impression $Imp(t,u,c_0,TL)$}
	\EndIf
\EndIf			
\EndWhile
\end{algorithmic}	
\end{algorithm}

It is worth making a few remarks about the Pre-Bid Randomization procedure. In Post-Bid randomization, logging a Test or Control impression occurs only {\em after} winning the auction, so every userID logged as Test actually sees an ad. By contrast, in Pre-Bid Randomization, a fraction $p$ of userIDs are logged as having "phantom" Control impressions ($C$) {\em before} submitting a bid to the exchange, and a bid is submitted only if the userID is among the $1-p$ userIDs in the Test group. For those in the Test group, if the DSP wins the auction, this results in logging a Test-Win ($TW$) impression, and otherwise a (phantom) Test-Loss ($TL$) impression. 

The meanings of the various types of impressions logged are:
\begin{itemize}
	\item Phantom control impression $Imp(t, u, c_0, C)$ indicates that at time $t$, a bid was about to be submitted for userID $u$ (on behalf of campaign $c_0$), but was not submitted because $u$ is a Control userID, i.e. $H(u,p) = C$.
	\item Test-Win impression $Imp(t, u, c_0, TW)$ indicates that at time $t$ the DSP {\em won} a bid submitted to the exchange on behalf of campaign $c_0$, and the Test userID $u$ is {\em exposed} to a campaign $c_0$ ad.
	\item Phantom Test-Loss impression $Imp(t, u, c_0, TL)$ indicates that at time $t$ the DSP \textit{lost} a bid submitted to the exchange on behalf of campaign $c_0$, and the Test userID $u$ is {\em not exposed} to a campaign $c_0$ ad.
\end{itemize}

\subsection{Causal Effect Computation}\label{sec-causal-pre-bid}

Let us now consider how we might estimate the causal effect ($ATT$) from the impression and event logs over a certain period. To keep the description simple, we will for now continue to make the assumption that each userID appears exactly once in the bid requests (Assumption \ref{assum-user-once}), and that there is a uniform distribution over userIDs (Assumption \ref{ass-unif}). The methodology and analysis for the (very realistic) recurring-userID case will be presented in Section \ref{sec-repeat}.

It will be useful to define a shorthand to denote various \textbf{populations} of userIDs: $C, TW, TL$ denote the set of userIDs that appear in the 3 respective types of impression logs. We also write $T$ to denote the union $TW \cup TL$. The sets $T$ and $C$ are clearly disjoint, and since we assumed each userID appears at most once in the logs, $TL$ and $TW$ are disjoint as well.

Recall that the $ATT$ is $E[ Y(1) | W=1] - E[Y(0) | W=1]$, i.e., the average observed response rate of exposed consumers minus the average (counterfactual) unobservable non-exposure potential response rate of exposed consumers. 
Under Post-Bid Randomization, our reasoning in estimating the $ATT$ was the following: 
\begin{enumerate}

	\item Assignment to Test/Control is exactly the same as exposure/non-exposure: a userID is exposed to an ad {\em if and only if} it is a Test userID,
	\item $E[Y(1) \,|\, W=1]$ is observable, as the response rate of Test (i.e., exposed) userIDs.
	\item Assignment to Test/Control is purely based on the userID $u$ (via the hash function $h_k(u)$), and the userID has nothing to do with its potential outcomes,
	\item Therefore exposure is \textit{ignorable} with respect to potential outcomes (Definition \ref{def-ignor}), or informally, the exposed (i.e. Test) and unexposed (i.e. Control) populations are "statistically equivalent" (with regard to potential response outcomes), and 
	\item This in turn implies (using Proposition \ref{prop-ignor}) that the counterfactual expectation $E[Y(0)|W=1]$ equals the observable response rate of Control userIDs, $E[Y(0)|W=0]$. 
\end{enumerate}

Can we do the same with Pre-Bid Randomization? The exposed population is precisely the set $TW$, so their average response rate $R_{TW}$ would be an estimate for the first term in the $ATT$, i.e., $E[ Y(1) | W=1]$. How do we estimate the counterfactual $E[Y(0) | W=1]$, i.e., the average response rate exposed consumers {\em would have had}, if they had {\em not} been exposed to the ad? 

If we examine the statements in the reasoning outlined above for the Post-Bid randomization setup, we see that:
\begin{itemize}
	\item Statement 1 does \textit{not} hold for Pre-Bid randomization: when a userID is assigned to Test, that userID is exposed to an ad \textit{only} if the submitted bid is won by the DSP, so assigning a userID to Test is \textit{not} equivalent to exposing that userID to an ad.
	\item Statement 2 is partially true in the Pre-Bid randomization scenario: $E[Y(1) \,|\, W=1]$ is observable, though it is not the response rate of Test userIDs; it is the response rate of the \textit{Test-Winner} userIDs $TW$.
	\item Statement 3 \textit{does} hold, since Test/Control assignment is still done via the deterministic hash function as in Post-Bid randomization, and \textit{this implies that Test/Control assignment is ignorable with respect to potential outcomes,} or informally, the Test and Control populations are "statistically equivalent".
	\item However since Test/Control assignment is not equivalent to exposure/non-exposure, the ignorability of Test/Control assignment does not translate to ignorability of exposure, so statement 4 cannot be made, i.e., \textit{exposure is not ignorable}.
	\item Thus, we cannot make Statement 5, i.e., we cannot say that the counterfactual expectation $E[Y(0) \,|\, W=1]$ equals the observable $E[Y(0) \,|\,W=0]$.
\end{itemize}

Thus the fundamental difficulty in the Pre-Bid randomization approach is that the exposure $W$ is not ignorable, or informally, the populations of exposed userIDs ($TW$) and unexposed userIDs ($C \cup TL$) are not "statistically equivalent". Indeed, even though the overall Test population $T = TW \cup TL$  is statistically equivalent to $C$ (i.e., Test/Control assignment is ignorable), the set $TW$ is the subset of $T$ consisting of userIDs in bids won in the auction, which could well be selecting for consumers who have a significantly higher or lower response rate (or Individual Causal Effect, ICE) than that of the overall Test population. In other words, there could be a significant {\em win bias} in the auction.

In fact, consumers that a DSP bids on and wins are expected to be systematically different from those it bids on and loses. This is at least in part due to the fact that bids lost are typically lost because some other buyer bid higher in the auction, presumably on the basis of some information unknown to the DSP in question. This might suggest, for example, that consumers bid on and lost might be systematically "more attractive" in some sense (i.e., with higher response rates). Thus, if not properly accounted for, win bias could actually produce a {\em negative} lift measurement, as the most responsive userIDs within the Test population are suppressed. Other systematic differences may also exist; the key point is that in general it cannot be assumed that responsiveness of consumers involved in won vs. lost bids is statistically equivalent.

In the following sections we present a rigorous framework to analyze the Pre-Bid Randomization scheme, and this will lead to a statistically sound methodology to estimate the $ATT$.

\section{Non-compliance and Auction Win Bias} \label{sec-win-bias}

Since in the Pre-Bid randomization setup, Test/Control assignment is no longer equivalent to exposure/non-exposure, we need to introduce a new random variable $Z_i$ that indicates whether userID $u_i$ has been {\bf assigned} to Test ($Z_i=1$) or Control ($Z_i=0$). As before, $W_i = 1$ and $W_i = 0$ denote whether the userID has been \textbf{exposed} or not, respectively. Equivalently, $W_i=1$ when the userID $u_i$ is in a \textbf{won bid}, and we informally refer to this userID as a \textbf{winner}. Similarly $W_i=0$ when $u_i$ is in a \textbf{lost bid}, and the userID is referred to as a \textbf{loser}. In the Post-Bid Randomization setup, $Z_i$ and $W_i$ are identical, but in the Pre-Bid Randomization scheme, there are three possibilities:
\begin{itemize}
	\item $Z_i = 0, W_i = 0$: userID $u_i$ assigned to Control and not exposed -- the set of all $u_i$ satisfying this are exactly the Control userIDs $C$,
	\item $Z_i = 1, W_i = 0$: userID $u_i$ assigned to Test but not exposed --  the set of all $u_i$ satisfying this are exactly the Test-Loss userIDs $TL$,
	\item $Z_i = 1, W_i = 1$: userID $u_i$ assigned to Test and exposed --  the set of all $u_i$ satisfying this are exactly the Test-Win userIDs $TW$,
\end{itemize}

The existence of userIDs $u_i$ for which $Z_i \neq W_i$ in general is called {\bf non-compliance} in the causality literature 
(see, e.g., \cite{Imbens1997a}). This terminology originates in randomized clinical tests of drug effectiveness, where certain patients assigned to a take a drug do not comply with their assignment. In our case there is only one type of non-compliance, exhibited by the $TL$ population, namely $Z_i = 1, W_i=0$. 

\subsection{User Types} \label{sec-user-types}
For a given Test userID $u_i$ we have $Z_i=1$, and whether or not they are exposed to an ad depends on whether the bid submitted involving $u_i$ wins the auction. This will in general depend on a number of factors, and since each userID is assumed to appear exactly once in the logs (Assumption \ref{assum-user-once}), we can think of each userID as encapsulating all of the factors that could possibly influence the winning of the auction. Then we can say that there are only two \textit{types} of userIDs: those whose bids would win if they were assigned to Test, and those whose bids would lose if they were assigned to Test. Thus, for a given type of userID, the exposure variable $W_i$ depends only on the test-assignment variable $Z_i$: $W_i = W_i(Z_i)$. For the analysis it will be useful to introduce a random variable $U_i$ that represents the \textit{user type} of userID $u_i$, with 2 possible values:
\begin{itemize}
	\item $U_i=1$ represents a "winner-type": $W_i(0)=0$, and $W_i(1) = 1$, i.e., bids involving this userID would be won if the UserID were assigned to Test,
	\item $U_i=0$ represents a "loser-type": $W_i(0)=0$ and $W_i(1)=0$, i.e., bids involving this userID would be lost if the userID were assigned to Test.
\end{itemize}
A more general treatment of user types can be found in \cite{Rubin_Causal_2005} or \cite{Chickering1996} but the above suffices for our context.

\subsection{Potential Outcomes and Response Rates} \label{sec-response-rates}

We define the potential outcome random variable $Y_i(Z_i)$ as the response (0 or 1) of a userID $u_i$ whose Test/Control assignment is $Z_i$. This is different from the earlier definition $Y_i(W_i)$ where $Y_i$ was written as a function of the \textit{exposure} status $W_i$ (of course in Post-Bid Randomization, there is no difference between $Y_i(W_i)$ and $Y_i(Z_i)$). 
From the above definition of user types it follows that once we fix the user type $U_i$, the exposure-status $W_i$ is determined only by the treatment-assignment variable $Z_i$. In subsequent analysis, it will be helpful to use the following intuitive shorthand notation to denote various populations, probabilities, and expectations (recall that all expectations and probabilities are calculated with respect to the uniform distribution over userIDs $i$). We use the superscript $'$ to indicate counterfactual averages that are not directly observable.

\begin{figure} \small \centering
\includegraphics[width=0.8\textwidth,angle=270]{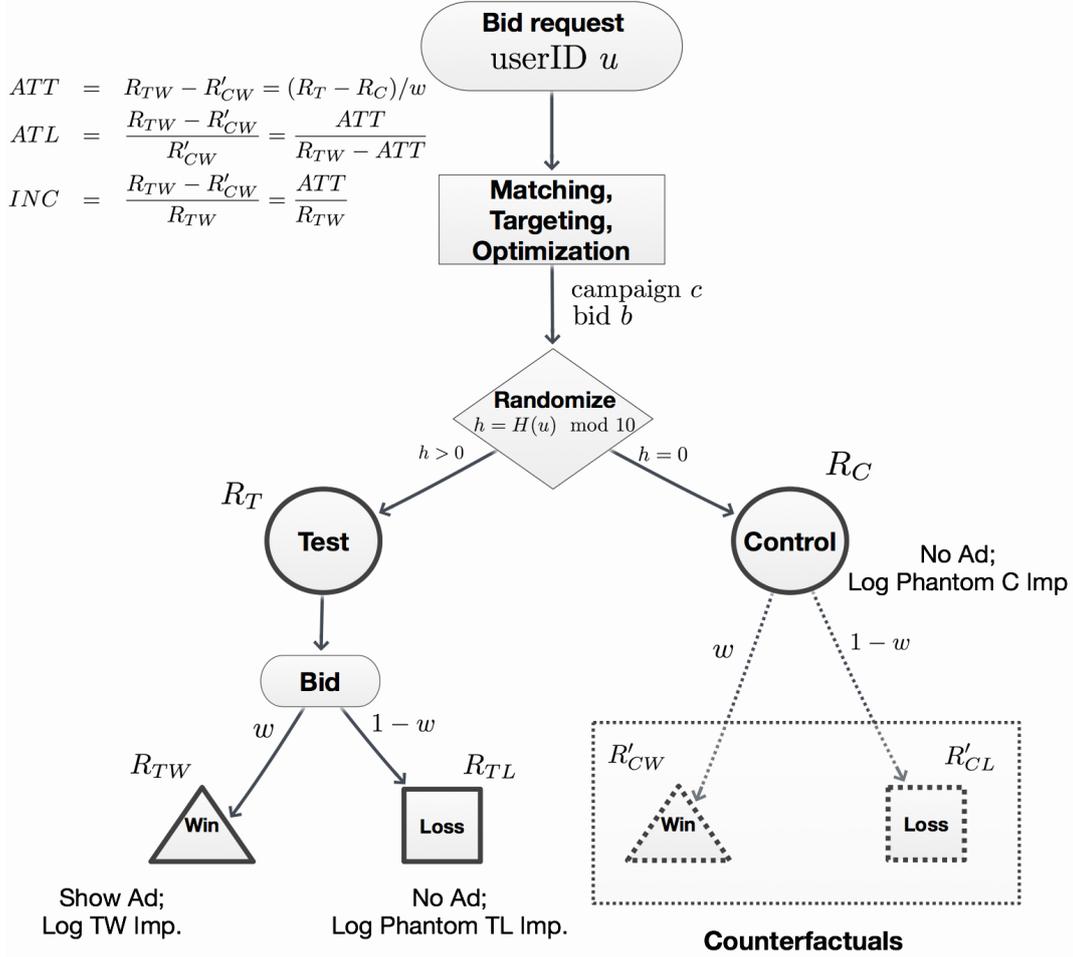}
\caption{\small Pre-bid randomization to estimate Causal Effect of ads. Test/Control assignment of a userID $u$ is based on the last digit of (the decimal representation of) $H(u)$ where $H$ is a uniform random hash function. 
The shapes (circles, squares, or triangles) labeled Test, Control, Win, Loss indicate different populations of userIDs, and are annotated with their corresponding response rates. 
Two populations are represented by the same shape when they are statistically equivalent, i.e., their response rates would be the same if their exposure status is the same. In particular, triangles correspond to "winner types" while squares correspond to "loser types." 
For example, the response rates $R_{TL}$ and $R'_{CL}$ are the same because they correspond to statistically equivalent (loser-type) populations, and neither is exposed to ads.
The solid shapes denote observable populations (and their response rates), while the two dotted shapes denote unobservable counterfactuals, whose response rates are $R'_{CW}$ and $R'_{CL}$. The win rate (i.e., the fraction of the Test population for which the submitted bid wins) is denoted $w$, and since the Test and Control populations are statistically equivalent, $w$ is also the expected fraction that \textit{would result in a win} if bids were submitted for the Control population. At the top left we summarize the main equations for lift computation under pre-bid randomization.
}
\label{fig-prebid-flow}
\end{figure}

\begin{itemize}
	\item $R_T = E[Y(1)]$, the overall observable average response rate of userIDs assigned to Test, regardless of their exposure status.
	\item $R_C = E[Y(0)]$, the overall observable average response rate of userIDs assigned to Control (none of these are exposed to ads).
	\item $R_{TW} = E[Y(1) | U=1]$, the observable average response rate of winner-type userIDs assigned to Test.
	\item $R_{TL} = E[Y(1) | U=0]$, the observable average response rate of loser-type userIDs assigned to Test.
	\item $R'_{CW} = E[Y(0) | U=1]$, the counterfactual average response rate of winner-type userIDs assigned to Control.
	\item $R'_{CL} = E[Y(0) | U=0]$, the counterfactual average response rate of loser-type userIDs assigned to Control.
	\item $w = P(U=1)$, the \textbf{win rate}, or the probability of a userID being a winner type. We show later how this can be computed from observable data.
\end{itemize}

Figure \ref{fig-prebid-flow} shows the various populations involved in the pre-bid randomization scheme.

\subsection{Causal Effect Definitions and Estimation}\label{sec-causal-est}

In this context, the earlier definitions of the $ICE$ (Individual Causal Effect), $ACE$ (Average Causal Effect) and $ATT$ are similar, except that we use the treatment-assignment variable $Z_i$ in place of $W_i$ as the argument of the $Y_i(.)$ and we use $U_i$ in place of $W_i$ as the conditioning variable in the conditional expectations.

\begin{definition}\label{def-ice-noncomp}\small
	\textbf{Individual Causal Effect ($ICE$) Under Non-Compliance:} The potential outcome of a userID $u_i$ if it were assigned to Test, minus the potential outcome if it were assigned to Control:
\begin{equation}
ICE_i = Y_i(1) - Y_i(0) \label{eq-ice-noncomp}	
\end{equation}
\end{definition}

Note that for a loser-type userID $i$ (i.e., $U_i = 0$), there is no exposure to the ad in either Test or Control, so $ICE_i = 0$.

\begin{definition}\label{def-ace-noncomp}\small
	\textbf{Average Causal Effect ($ACE$) Under Non-Compliance:}
\begin{equation}
ACE = E_i[ICE_i] = E[ Y(1) - Y(0) ] = R_T - R_C \label{eq-ace-noncomp}	
\end{equation}
\end{definition}

Note that $R_T = E[ Y(1)]$ is the overall average response rate of the Test group, including winner types ("compliers") and loser types ("non-compliers"), and $R_C = E[Y(0)]$ is the overall average response rate of the Control group. For this reason the $ACE$ under non-compliance is also called the \textbf{Intent-To-Treat} ($ITT$) effect: it represents the average response-rate differential due to the \textit{intention} to expose the Test group to ads, regardless of who in the Test group is actually exposed. Since both response rates are observable, it is straightforward to obtain an unbiased point estimate of the $ACE$.

\begin{definition}\label{def-att-noncomp}\small
	\textbf{Average Treatment Effect on Treated ($ATT$) Under Non-Compliance:}
\begin{align}
ATT & =  E_i[ ICE_i | U_i = 1]   \nonumber \\
    & =  E_i [Y_i(1) - Y_i(0) | U_i = 1] 	\nonumber \\
	& =  R_{TW} - R'_{CW}  \label{eq-att-noncomp}	
\end{align}

\end{definition}

This definition of $ATT$ conditions on exactly the right population: the winner types, with $U=1$. This is precisely the population that \textit{would} be exposed to ads \textit{if} they were assigned to Test, and the average response rate would be $R_{TW} = E[Y(1)| U=1]$. Similarly this population would \textit{not} be exposed if assigned to Control, and the average response rate would be $R'_{CW} = E[Y(0)|U=1]$. There is, however, a problem: the user type $U$ is \textit{not observable} for Control userIDs, so $R'_{CW}$ is not directly observable! Nevertheless, we can compute the $ATT$, thanks to a series of known properties and results (see, e.g., \cite{Little_Causal_2000}):

We begin with a straightforward property:
\begin{proposition} \label{prop-rtl-rcl}\small
Under Pre-Bid Randomization, the average response rate of loser-type userIDs assigned to Test and Control are the same, i.e. $E[Y(1) | U=0] = E[Y(0) | U=0]$, or $R_{TL} = R'_{CL}$.
\begin{proof}
This follows from the fact that (a) loser-type userIDs are not exposed to ads, whether they are assigned to Test or Control, and (b) Test/Control assignment is purely random and hence ignorable, i.e., it has no influence on potential outcomes under non-exposure.
\end{proof}

\end{proposition}

\begin{lemma}\small
In the Pre-Bid Randomization setting, an unbiased estimate of the probability that a userID is a winner type, or the win rate $w = P(U=1)$, is given by the fraction of Test userIDs that are exposed. 
\begin{proof}\small
Note that $P(U=1) = E(U)$, the average value of the user type $U$ across the population. Since Test/Control assignment is purely random, and hence ignorable, we can write $E(U) = E(U | Z=1)$. But in the Test population, the userIDs with $U=1$ are precisely those with $W=1$, so $E(U|Z=1) = E(W|Z=1)$, and an unbiased estimator of $E(W|Z=1)$ is the fraction of Test userIDs that are exposed.
\end{proof}
\end{lemma}

\begin{lemma} \label{lem-att-prebid}\small
	In the Pre-Bid Randomization setting, the $ATT$ is the ratio of two observable averages:
\begin{equation}
ATT = ACE/P(U=1) = ACE/w = (R_T - R_C)/w, 	\label{eq-att-prebid}
\end{equation}
where $w$ is the win rate, or equivalently, 
\begin{equation}
E[Y(1) - Y(0) | U=1] = E[Y(1) - Y(0)]/P(U=1), 	\label{eq-att-prebid-expec}
\end{equation}

\begin{proof}\small
We can decompose the $ACE$ (Def. \ref{def-ace-noncomp}) as a sum of two terms:
\begin{eqnarray}
ACE &= & E[ Y(1) - Y(0) ] \; = \; R_T - R_C \\
	&= & P(U=1)\cdot (R_{TW} - R'_{CW}) \; +\; \\
	&  &  P(U=0)\cdot (R_{TL} - R'_{CL})  \\
    &= & w \cdot ATT \; + \; (1-w) \cdot (R_{TL} - R'_{CL}),
\end{eqnarray}
where the second term in the last expression vanishes due to Prop. \ref{prop-rtl-rcl}, implying the result.
\end{proof}
\end{lemma}

Advertisers are often interested in a variant of the $ATT$ that measures the change in response rate relative to the baseline unexposed response rate:

\begin{definition}\label{def-lift}\small
	\textbf{Average Treatment Lift (ATL):} 
\begin{eqnarray}
ATL &= & (R_{TW} - R'_{CW})/R'_{CW}\\
    &= & ATT/R'_{CW}\\
    &= & ATT/(R_{TW} - ATT) \label{eq-atl}	
\end{eqnarray}
\end{definition}

The $ATL$ will be informally referred to as the \textit{Causal Lift}, or \textit{Ad Lift}. An alternative way of expressing the relative (causal) impact of an ad campaign is the \textbf{Incrementality}, which is similar to the $ATL$ except that it is expressed relative to $R_{TW}$, or the response rate of the exposed Test population:

\begin{definition}\label{def-inc}\small
	\textbf{Incrementality (INC):} 
\begin{eqnarray}
INC &= & (R_{TW} - R'_{CW})/R_{TW}\\
    &= & ATT/R_{TW} 	\label{eq-inc}	
\end{eqnarray}
\end{definition}

In other words, $INC$ answers the question, \textit{what fraction of the exposed Test response rate is caused by the ad campaign}, and is therefore necessarily bounded by 1.0 (unlike the $ATL$, which has no upper bound). From the above definitions we see that both the $ATL$ and $INC$ are straightforward functions of $R_{TW}$ and $ATT$.

\subsection{Measuring Response Rates from Impression and Event Logs}
We use the simple method described in Section \ref{sub-causal-post-bid} to estimate the response rates   
$R_T, R_C, R_{TW}, R_{TL}$, which are the response rates of the populations $T,C,TW,TL$, respectively. 
We note that for our methodology it is crucial that the DSP is able to log bid opportunities, and the effectiveness of this technique is proportional to a DSP's ability to view and log a large fraction of all available bid opportunities.
 
\subsection{Assumptions about the Test and Control Groups}

For this discussion, suppose advertiser $A$ is measuring the effectiveness of campaign $c_0$ on the DSP $D$. For brevity, we will say a campaign $c_1$ is \textit{related} to campaign $c_0$, if exposure to $c_1$ has an influence on a consumer's likelihood of performing campaign $c_0$'s desired outcome. In particular, campaign $c_0$ is of course related to itself. 

In general, it is possible that consumers in both the Test and Control groups of campaign $c_0$ are exposed to ads from related campaigns by Advertiser $A$, either from the DSP $D$ in question, from another DSP executing ads for $A$ in digital channels, or from advertising by $A$ in non-digital channels (e.g., traditional TV, print, radio, etc.). Since Test/Control assignment is specifically with respect to campaign $c_0$, exposure from such related campaigns is effectively random, in the sense that userIDs in Test and Control groups (and in particular the winner-types in these groups) for campaign $c_0$ are equally likely to be exposed to those related campaigns. Thus the methodology presented here for measuring $ATL$ and $INC$ remains intact, but with two qualifications:

\begin{itemize}
	\item The interpretation of Control as a baseline corresponding to "no ad exposure" and Ad Lift relative to that "zero" baseline is no longer accurate. Instead, both the Control and Test groups are affected, in identical fashion, by related campaigns that effectively introduce a non-zero baseline of response against which Ad Lift is measured, i.e., they are simply boosting the baseline. 
	\item The interpretation of Ad Lift and Incrementality is then understood to be not "the causal impact of campaign $c_0$ in the absence of all other advertising" but rather "the causal impact of campaign $c_0$ against the background of all related advertising." Presumably, the latter is smaller in magnitude than the former\footnote{The causal impact of campaign $c_0$ in the absence of all other advertising isn't knowable in practice, as it would require an experiment that involved shutting off all related advertising in all channels, waiting for its lagged effect to subside, and running only a single campaign.}, but the latter is all that matters as it represents the real-world impact of campaign $c_0$. The methodology described here correctly quantifies the corresonding Ad Lift, and is the correct quantity against which the advertiser should measure the impact if campaign $c_0$. 
\end{itemize}

In fact, the effective baseline for $c_0$ isn't only affected by related digital and non-digital ad campaigns by Advertiser $A$. The likelihood of consumers to exhibit the desired outcome for campaign $c_0$ is also impacted by  related advertising from $A$'s competitors, by so-called earned media (i.e., what consumers are saying about $A$'s products or services in social media), by seasonality, macro-economic conditions, and myriad other factors. Again, as long as these factors do not systematically affect Test and Control groups for $c_0$ in different ways, they merely serve to determine the (non-zero) Control baseline, and the methodology here pertains. 

Given that the assignment of userIDs to Test and Control is effectively random (as discussed in Section \ref{sec-psa}, based on the hash function $h_k(u)$ as applied to userIDs within DSP $D$), it is reasonable to assume that such factors pose no issue, as they would impact Test and Control in statistically equivalent fashion. It is difficult (though perhaps not impossible) to imagine how external factors, including related campaigns, could result in different \textit{systematic} variation of response between Test and Control (though they can certainly introduce \textit{statistical} variation, i.e., noise). However, given that the Test/Control assigment happens within the DSP, it is possible that some issue internal to the DSP $D$ in question could be causing related campaigns executed by $D$ to introduce some systematic effect, contaminating the results\footnote{Even more problematic would be if the advertiser were to duplicate the {\em exact same} campaign, either on DSP $D$ or another DSP. This could pose additional complications, but we do not address those here. In practice, we do not observe advertisers engaging in this behavior; they may run different campaigns with different DSPs, or run disjoint parts of the same campaign on different DSPs (e.g., running display ads with one DSP and mobile ads with another), but we have not observed many instances of the same advertiser running the exact same campaign on multiple DSPs.}.

To mitigate the impact of these issues, we recommend the following practices be followed when measuring the effectiveness of campaign $c_0$ of advertiser $A$ on DSP $D$:

\begin{enumerate}
		\item Advertiser $A$ should run \textit{all} campaigns related to $c_0$  on the \textit{same} DSP $D$, and 
		\item If advertiser $A$ is running campaigns related to $c_0$ on DSP $D$, then either:
	\begin{itemize}
		\item these are also set up to measure effectiveness, with the same Test and Control assignment function $H(u,p)$, such that a userID is in the Control group for all campaigns related to $c_0$, or
		\item the other campaigns related to $c_0$ target a set of userIDs that is disjoint from the set of userIDs targeted by $c_0$. 
	\end{itemize}
\end{enumerate}

These practices, which should be easy to implement for a DSP alraedy capable of executing Ad Lift measurement via Pre-Bid Randomization, will avoid systematic contamination of the baseline.  

\section{Gibbs Sampling for Confidence Bounds} \label{sec-gibbs}

Until now we have been only concerned with \textit{point estimates} of various types of causal effects. However a point estimate $x$ by itself has little value in the absence of a confidence-interval $(x_0,x_1)$ that indicates, for example, that given the observed data, we are 90\% confident that the true effect is between $x_0$ and $x_1$.

The difficulty in obtaining confidence bounds depends on which type of causal effect we are considering. 
We are mainly interested in the $ATT$ (Def. \ref{def-att-noncomp}) and the 
$ATL$ (the Average Treatment Lift, Def. \ref{def-lift}). 
The equations for the point estimates of these effects (Eqs. \ref{eq-att-noncomp} and \ref{eq-atl} respectively) involve various uncertain response rates and the win rate $w$. 
For example \cite{Imbens1997a} presents an estimate of the standard deviation of the posterior distribution of the $ATT$ under a large-sample normal approximation. 
However this approach cannot be extended to the $ATL$.

We have developed a methodology to compute confidence bounds for both metrics (or any other similarly-defined metric) based on a simple Gibbs sampling scheme. Such schemes were proposed for example by \cite{Chickering1996} and \cite{Barajas2012}, but we believe our methodology, as well as its presentation, is simpler and more intuitive to understand. Moreover, unlike the confidence-interval estimation of \cite{Imbens1997a}, our methodology is relatively assumption-free, works for non-large samples, and does not rely on a normal approximation.

\subsection{Basic Ideas}

For an introduction to Gibbs Sampling we refer the interested reader to \cite{resnik2010gibbs} and the references therein, but here we will give a very informal overview tailored to our purposes. Gibbs Sampling is a special case of a general  method called Markov-Chain Monte-Carlo (MCMC) sampling. MCMC techniques are applicable when we have a Bayesian \textit{generative model} of our data $D$, parameterized by some unknown parameter vector $\pi$, and we want to generate random draws $\pi_0, \pi_1, \ldots$ from the joint posterior distribution of the parameters $\pi$ given the observed data. There are a couple of reasons why one might want to do this: (a) compute an expectation (or average) of some scalar function $f(\pi)$ with respect to the distribution of $\pi$: $E_\pi(f(\pi))$, and (b) for some scalar function $f(\pi)$, compute a range $[f_1, f_2]$ such that 90\% of the probability mass of $f(\pi)$ lies between $f_1$ and $f_2$.

The reason the technique is called a Markov-Chain Monte Carlo method is two-fold: (a) the random draws $\pi_0, \pi_1, \ldots,$ are generated iteratively in such a way that the next draw depends only on the outcome of the preceding draw, such that one can imagine these draws as states in a Markov Chain and the iterative process is performing a "walk" on the Markov Chain, and (b) the applications involve computing quantities (expectations, confidence intervals, etc.) based on simulating (i.e., drawing from) a joint distribution of some stochastic variables (hence the "Monte-Carlo" in the name). 

The crucial property for an MCMC process to be useful is that after some initial number of "burn-in" iterations, the distribution of the draws $\{\pi_i\}$ converges to the true posterior distribution of the parameters $\pi$ given the observed data $D$. It is this property that enables accurate computation of averages, confidence-bounds, etc. Gibbs Sampling is a specific way of performing an MCMC simulation where the parameter vector $\pi$ can be partitioned into $k \geq 2$ parts, and at each iteration, rather than generating the entire new vector $\pi$, each part is generated separately conditional on the latest values of all the other parts.

\subsection{Observed Counts, Hidden Counts, and Parameters}
For our present purpose of computing confidence intervals for $ATT$ and $ATL$, we have the following components in order to construct a Gibbs sampling scheme (it will be helpful to consult Figure \ref{fig-gibbs} for the rest of this section):
\begin{itemize}
	\item The observed data $D$ consists of counts of responders and non-responders among the 3 populations $TW, TL, C$. In general we will write $X_1$, $X_0$ to denote the number of responders and non-responders respectively in population $X$. Thus we have 6 observed counts: $TW_0,TW_1, TL_0, TL_1, C_0,$ and $C_1$
	\item The four \textit{unobserved} counts are the number of responders and non-responders in the counterfactual subsets $CW$, $CL$, where $CW$ is the set of winner types in Control, i.e., the set of Control userIDs who \textit{would have won} if they were assigned assigned to Test, and 
$CL$ is the set of loser types in the Control group. More precisely, 
$CW = \{u_i | U_i=1, Z_i=0\}$, and 
$CL = \{u_i | U_i=0, Z_i=0\}$. The four unobserved counts are $CW_0, CW_1, CL_0, CL_1$. Note the constraints
\begin{equation} \label{eq-cw-cl}
	CW_0 + CL_0 = C_0; \;\;\; CW_1 + CL_1 = C_1
\end{equation}
	\item The response rates $R_{TW}, R_{CW}, R_{TL}=R_{CL}=R_L$ 
and the win rate $w$ are treated as \textit{unknown parameters} 
for which we want to simulate the posterior joint distribution conditional on the data. Recall the equality of $R_{TL}$ and $R_{CL}$ from Prop. \ref{prop-rtl-rcl}. We drop the $'$ superscript on $R_{CW}$ here since we are treating all the response rates as unknowns
\end{itemize}

\begin{figure} \small \centering
\includegraphics[width=0.8\textwidth,angle=270]{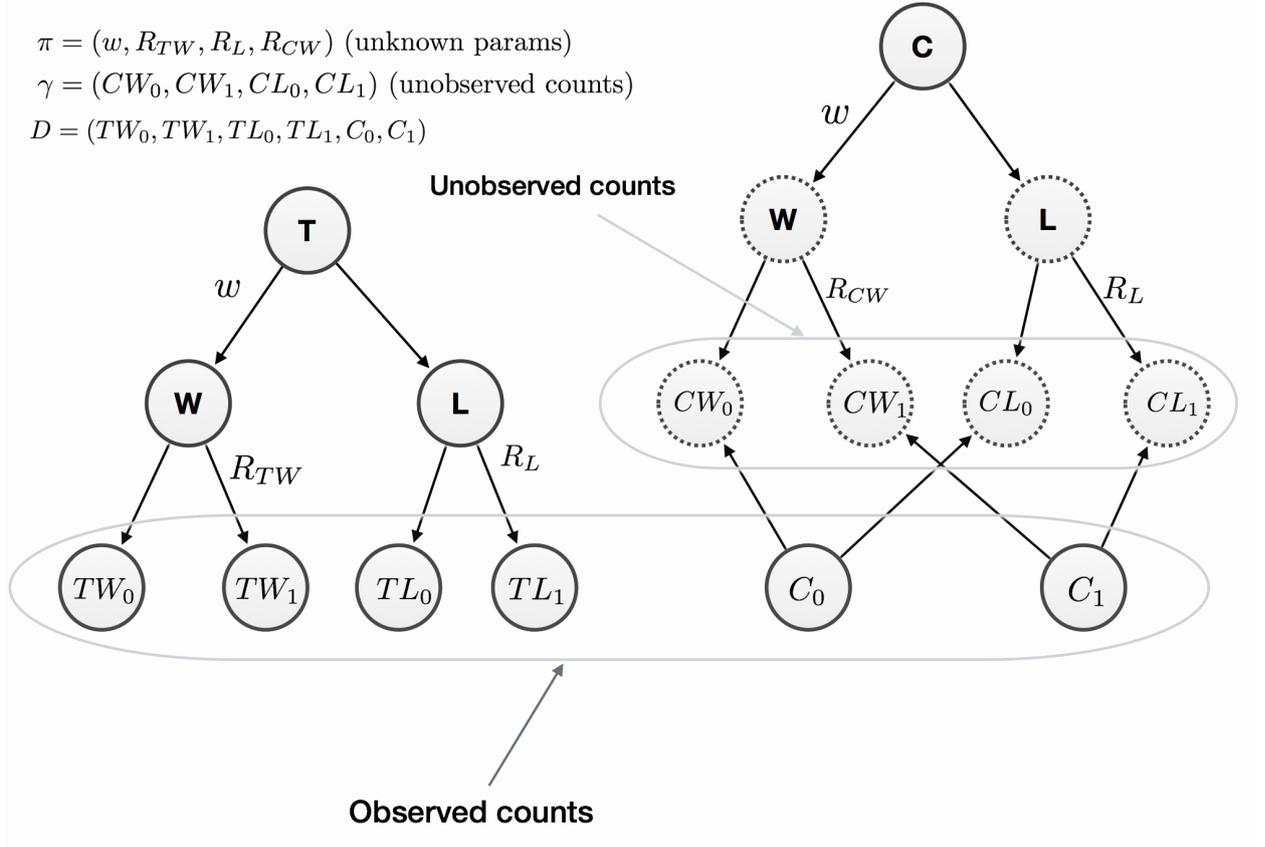}
\caption{\small Random variables and observed counts in the Gibbs Sampling procedure. The nodes labeled $T, C$ denote the Test and Control populations respectively, and the $W, L$ nodes under $T, C$ denote the winner-type and loser-type sub-populations. The other nodes contain variables denoting counts using a mnemonic notation for the population, and the subscript 1 indicates converters, and 0 indicates non-converters. For example $TW_0$ is the number of non-converters among the Test-Win population. Solid circles denote observable counts, while dotted circles denote unobservable counts. Note that although the counts $CW_i, CL_i$ are not observable, the sum $CW_i + CL_i = C_i$ is observable, since it is the total number of converters (for $i=1$) or non-converters (for $i=0$) in the Control group. The arrows are labeled by the probability that an individual userID in the parent-node population falls in the sub-population represented by the child node. For example the arrow from $T$ to $W$ is labeled $w$ to show that the probability that a Test-group individual is a winner type is $w$ (the win rate), and the arrow from $W$ to $TW_0$ is labeled $R_{TW}$ to indicate that the probability of an individual in the $TW$ (Test-Winner) group converts is $R_{TW}$, the Test-Winner response rate. The complementary probabilities $1-w$, etc. are not shown.}
\label{fig-gibbs}
\end{figure}

\subsection{Data Generating Model}

Given the parameters $w, R_{TW}, R_{TL}, R_L$, let us specify the process by which the observed and un-observable counts above are generated: this will constitute the \textit{generative model} for the data, and we will use this model when simulating the draws from the posterior distribution of the parameters given the data.
\newcommand{\bindist}{\text{Binom}}
The total number of winner-type userIDs is $N_W$ = $CW_0 + CW_1 + TW_1 + TW_0$, and the total number of loser-type userIDs is 
$N_L = CL_0 + CL_1 + TL_0 + TL_1$, so $N_W$ is a draw from the Binomial distribution with total trials $N_W + N_L$ and probability of success = $w$, or written more succinctly:
\begin{equation}
N_W \sim \text{Binom}(N_W + N_L, w) \label{eq-like-w}
\end{equation}
Similarly we can write:
\begin{eqnarray}
CW_1 &\sim & \text{Binom}(CW_0 + CW_1, \; R_{CW})\\ \label{eq-like-rcw}
TW_1 &\sim &\text{Binom}(TW_0 + TW_1, \; R_{TW}) \\ \label{eq-like-rtw}
CL_1 + TL_1  & \sim & \text{Binom}(CL_0 + TL_0 + CL_1 + TL_1, \; R_L) \label{eq-like-rl}
\end{eqnarray}	

The generative model for the the unobserved counts $CW_x, CL_x$ ($x \in \{0,1\})$ requires a bit more thought. For example, to arrive at the generative model for $CW_1$ we use the fact that $CW_1$ and $CL_1$ add up to the observable count $C_1$ (Eq. \ref{eq-cw-cl}). This means that $CW_1$ is obtained by a binomial draw from $C_1$ trials, with success probability 

\begin{eqnarray}
\lefteqn{P(\text{winner-type} \;|\; \text{control responder}) }\\
    & = & P( U = 1 | Y(0) = 1) \\
	& = & P( U = 1, Y(0) = 1)/P(Y(0) = 1) \\
	& = & \frac{P( U = 1) P (Y(0) = 1 | U = 1)}
	      {P( U = 1) P (Y(0) = 1 | U = 1) + 
	   	   P( U = 0 ) P (Y(0) = 1 | U = 0)}\\
	& = & \frac{w R_{CW}}{w R_{CW} + (1-w)R_L}.
\end{eqnarray}

Thus we can write:
\begin{equation} \label{eq-draw-cw-cl}
\begin{split}
CW_1 &\sim \text{Binom}(C_1, wR_{CW}/[ wR_{CW} + (1-w)R_L]) \\	
CL_1 &= C_1 - CW1\\
CW_0 &\sim  \text{Binom}(C_0, (1-w)R_{L}/[ (1-w)R_L + wR_{CW}]) \\
CL_0 & = C_0 - CW_0	
\end{split}	
\end{equation}

\subsection{Posterior Parameter Distributions}
The final piece needed to set up a Gibbs sampling process is the specification of the posterior distributions of the parameters $w, R_{TW}, R_{CW}, R_L$.
Following a Bayesian approach, we first specify prior distributions of each of these parameters. To model minimal prior knowledge about these parameters we use a $\text{Beta}(1,1)$ prior, which is equivalent to a uniform distribution over [0,1]. 

We note that each of the Binomial distributions in Eqs. \ref{eq-like-w} - \ref{eq-like-rl} corresponds to a \textit{likelihood} of the data given the respective parameter, e.g., from Eq. \ref{eq-like-w} the likelihood is the probability of selecting $N_W$ winner types from a total of $N_W + N_L$ trials where each has a probability $w$ of being a winner type:
\begin{eqnarray}
P(N_W \; | \; 	N_W + N_L, w) = {N_W + N_L \choose N_W} w^{N_W} (1-w)^{N_L}
\end{eqnarray}

To avoid repeating these expressions, we use the notation 
$\mathcal{L}_{bin}(k; n, p)$
 to denote the Binomial likelihood of $k$ successes out of $n$ trials, with success probability $p$, or ${n \choose k}p^k (1-p)^{n-k}$. We also use the following property: 
 
\begin{proposition} \label{prop-conj}\small
The Beta distribution is a \textit{conjugate prior} to the Binomial likelihood, i.e., if the prior distribution on a probability $p$ is a Beta($\alpha$, $\beta$) distribution, and the likelihood is a binomial likelihood $\mathcal{L}_{bin}(k; n,p)$, then the posterior of $p$ is also a Beta distribution, specifically with parameters ($\alpha + k, \beta + n-k$):
\begin{eqnarray}
P(p | n, k) &\propto  \mathcal{L}_{bin}(k; n,p) \text{Beta}(\alpha, \beta)\\	& \propto  \text{Beta}(\alpha + k, \beta + n-k)
\end{eqnarray}
\end{proposition}
	
Thus from Eqs. \ref{eq-like-w} - \ref{eq-like-rl} we can write the following posterior distributions (note in our case $\alpha=\beta=1$):
\newcommand{\betadist}{\text{Beta}}
\begin{equation} \label{eq-post}
\begin{split}
w &\sim  \betadist(1 + CW_1 + CW_0 + TW_1 + TW_0,
					  1 + CL_1 + CL_0 + TL_1 + TL_0) \\
R_{TW} & \sim 	\betadist(1 + TW_1, 1 + TW_0) \\
R_{CW} & \sim  \betadist(1 + CW_1, 1 + CW_0) \\
R_L & \sim  \betadist(1 + CL_1 + TL_1, 1 + CL_0 + TL_0)	
\end{split}
\end{equation}

\subsection{Gibbs Sampling Procedure}

We are now ready to specify the actual Gibbs sampling procedure. For brevity we will let $\pi$ denote the vector of the four unknown probability parameters $w,R_{TW}, R_{CW}, R_L$, and let $\gamma$ denote the vector of the four unobservable counts (see Eq. \ref{eq-cw-cl}) $CW_1, CW_0, CL_1, CL_0$. Our goal is to iteratively generate a sequence of realizations of $[\pi,\gamma]$ such that after some initial "burn-in" period (typically around 1000 iterations), these realizations represent draws from the true joint distribution of $[\pi, \gamma]$. The sequence of realizations of $[\pi,\gamma]$ will be denoted $[\pi^0, \gamma^0], [\pi^1, \gamma^1], \ldots$, and we use the same superscript notation to refer to iterative realizations of the components of $\pi$ and $\gamma$.

We first initialize $\pi$ to $\pi^0$ using reasonable estimates from the observed counts:
\begin{equation}
\begin{split}
w^0 &= (TW_0 + TW_1)/(TW_0 + TW_1 + TL_0 + TL_1) \\
R^0_{TW} &= TW_1/(TW_0 + TW_1) \\
R^0_{CW} &= TW_1/(TW_0 + TW_1) \\
R^0_L &= TL_1/(TL_0 + TL_1),
\end{split}
\end{equation}
where we note that we initialized $R_{CW}$ using the observed \textit{Test}-Winners response rate.

Next we update $[\pi,\gamma]$ iteratively, where the $i$'th iteration ($i = 0,1, 2, \ldots$) consists of two steps:
\begin{enumerate}
	\item Generate $\gamma^{i+1}$ given $\pi^i$, using equations \ref{eq-draw-cw-cl}. For example to generate the count $CW^{i+1}_1$, we make a random draw from $\bindist(C_1, p)$ where the success probability $p$ is 
$$
w^i R^i_{CW}/[ w^i R^i_{CW} + (1-w^i)R^i_L])
$$

	\item Generate $\pi^{i+1}$ given $\gamma^{i+1}$, using equations \ref{eq-post}. For example to generate $R^{i+1}_{TW}$ we make a random draw from $\betadist(1+TW^{i+1}_1, 1 + TW^{i+1}_0)$.

\end{enumerate}

Once we pass a suitable "burn-in" number of iterations $N$, at each of the $k$ subsequent iterations $i = N+1, N+2,\ldots,N+k$, we calculate the metrics $ATT^i, ATL^i$, which are simple functions of the probability-parametres $\pi^i$ at iteration $i$ 
(see Eqs. \ref{eq-att-noncomp}, \ref{eq-atl}). 
We typically use $N=1000$, and $k = 2000$. The 90\% confidence bounds on $ATT$ are then given by the 5th and 95th percentiles of the collected values $\{ATT^i\}$, and similarly for the confidence bounds of $ATL$. Note that we can use the averages of $\{ATT^i\}$ and $\{ATL^i\}$ as an alternative to the point-estimates in Eqs. \ref{eq-att-noncomp} and \ref{eq-atl}, and also as a sanity check on the Gibbs sampling.

\section{Complication: Recurring UserIDs}  \label{sec-repeat}

Thus far in this paper we have assumed that a userID never occurs more than once in the logs data used to measure Causal Lift. 
In reality, a given userID may (and typically does) appear multiple times in the bid-request stream, as the consumer engages with different websites or mobile apps. In this section we will extend our methodology to handle this situation. 

To see the difficulty caused by recurring userIDs, consider a naive approach where we treat each bid opportunity for a given userID as a distinct unit, pretending they are all different userIDs and carrying out the analysis presented earlier. Recall that our methodology described in Sec. \ref{sec-causal-est} hinges on identifying the \textit{winner} and \textit{loser} subsets of the Test group, i.e., the $TW$ and $TL$ populations. 
For a given Test userID, we might say that a bid opportunity is in $TW$ if the bid opportunity results in a win, and is in $TL$ otherwise. This would be problematic because if an individual consumer appeared with a certain userID $u$ in a won bid opportunity at time $t$ (and hence was exposed to the ad campaign), and then re-appeared in a lost bid opportunity (with the same userID $u$ at time $t + \delta$, this consumer could still very well be "under the influence" of the ad exposure from time $t$.
This means the potential outcomes $Y(0), Y(1)$ subsequent to the second bid opportunity would not be those that one would expect from an unexposed consumer. 
In the extreme case, if the gap $\delta$ between the first and second bid opportunities is close to zero, it is quite likely that the "if-unexposed potential outcome" $Y(0)$ is identical to $Y(1)$, i.e., \textit{regardless of whether the second bid opportunity results in a win or not}, the two potential outcomes are \textit{identical}. This extreme case is actually quite likely in practice, corresponding for example to consecutive bid opportunities that might occur as a consumer clicks through from one page of a website to another.

Thus, if we treat each bid opportunity as a unit, there will be interference among units in the sense that the potential outcomes of a unit would depend on whether or not other units were exposed. This violates a fundamental requirement that underlies the potential outcomes framework, called the Stable Unit Treatment Value Assumption (SUTVA), \cite{Rubin_Causal_2005}. This assumption states, roughly speaking, that the potential outcomes of a unit do not depend on the exposure-status of another unit (and so the SUTVA is informally referred to as the "no-interference" assumption). This assumption was not explicitly called out in our analysis thus far because we have been assuming that each userID occurs exactly once, and hence it is very reasonable to assume no interference between units (i.e. userIDs) in that scenario.
To avoid violating the SUTVA, instead of treating each individual bid opportunity of a userID as a unit, we make the following modification to our methodology:

\begin{modification}[Treat UserID as a Unit]
Treat the userID itself as a unit, which effectively means that each "unit" now represents the set of \textit{all} bid opportunities that have a given userID. 
\end{modification}

With this modification, units now correspond to individual consumers, who may each be the subject of multiple bid opportunities. Clearly now there will be no interference between units, \textit{assuming each distinct userID corresponds to a distinct consumer} (the next section will consider the case where this assumption is not valid): ignoring social-network effects (which are generally not predicated on advertising), one consumer's ad exposure will not affect another's potential outcomes. In fact, defining units in this way is quite natural from an advertiser's perspective: they are primarily interested in measuring the causal impact of their campaigns at the level of \textit{unique userIDs} (i.e., their customers and potential customers) rather than individual bid opportunities.

In order to apply our methodology, we now need to specify how to define whether a unit is a Test-Winner ($TW$) or Test-Loser ($TL$) (the definition of $C$, or control, is clear because Test/Control assignment is based on the userID):

\begin{modification}[Definition of Test-Winner and Test-Loser UserIDs]
A Test userID $u$ is considered to be a Test-Winner if \textit{at least one} of the bid opportunities involving $u$ is a \textit{win}, otherwise it is considered a Test-Loser.
\end{modification}

In the case of non-recurring userIDs, determining whether not a bid opportunity is associated with a conversion is a simple matter: we simply check whether there is a conversion by the userID within the campaign-designated post-view (PV) window after the bid opportunity. 
When userIDs occur more than once, we define the response as follows:

\begin{modification}[Attribution of Conversions to UserIDs] \label{mod-response}
	For a given userID $u$, a conversion by $u$ is \textit{attributable} to $u$ if there is a bid opportunity for $u$ within the campaign-designated post-view (PV) window prior to the conversion event. The \textit{response} $Y_u$ of the userID is then defined as the total number of conversions attributable to $u$.
\end{modification}

This definition of the response random variable $Y_u$ allows us to extend the definitions of potential outcomes ($Y_u(0)$ and $Y_u(1)$) and the various response rates introduced in Section \ref{sec-response-rates}, to the case of recurring userIDs. For example, the response rate $R_{TW} = E[Y(1) | U=1]$ is simply the total number of conversions attributed to Test-Winner userIDs (as in Modification \ref{mod-response}), divided by the number of unique Test-Winner userIDs. With these definitions, it is easy to apply the methodology to compute $ATT$ and $ATL$ in Section \ref{sec-causal-est}.

Note that the introduction of these modifications means that different consumers in the Test group may be subject to different degrees of ad exposure (i.e., because the DSP may win only one bid opportunity against some Test-Winner userIDs, and multiple bid opportunities against other Test-Winner userIDs.). In theory, consumer response to advertising should be some function of the degree (i.e., frequency) of exposure. For the current purposes, the shape of this function is irrelevant, in the sense that the response rates in question (i.e., $R_{T}$, $R_{TW}$, etc.) are simply aggregates over the corresponding populations. The degree of exposure that drove the response rates does not affect the calculation of Causal Lift. What would be affected is the measured return on investment (ROI) of the campaign, which looks at how much spend was required to generated the observed lift, as higher frequency of exposure would entail higher cost and vice-versa. ROI considerations will be visited in later work; the current approach seeks to quantify the Causal Lift independent of the spend that was required to produce it.

\section{Complication: ID Contamination} \label{sec-contam}

The methodology thus far has been based on the \textit{userID} as the unit of study, and essentially calculates the \textit{userID-level} Causal Lift (even when a userID occurs multiple times). However, as mentioned in the Introduction, a given consumer may be associated with multiple userIDs at any given time, and these often change over time. 
In reality, though we have been using the term "userID," these identifiers actually correspond to browsers or devices, rather than the actual consumers using them. Perhaps the most well-known type of userID is a browser cookie: a piece of data generated by code on a website, that is exchanged between the website and the browser and serves to identify the individual instance of the web browser being used on a particular device. Increasingly, cookies are periodically deleted, either automatically by the browser (e.g., after each browsing session) or manually by the user (e.g., see \cite{comscore-2007}).
\footnote{The situation is even worse with respect to browsers such as Safari, which default to not accepting cookies to begin with and therefore do not  allow for recognition of the device at all. In such cases, all cookie-based response rates would effectively be zero, as conversions cannot be associated with bid opportunities, and Ad Lift becomes indeterminate.}
Thus, cookies do not necessarily persist across different observations of the same consumer, such that a given consumer could appear with a distinct (cookie-based) userID on different bid opportunities. In the context of pre-bid randomization, this would generally result in some of the consumer's userIDs being assigned to Test and others to Control, and we refer to this as \textit{cookie contamination.}
\footnote{Strictly speaking, the issues of cookie deletion and non-acceptance are most pronounced for {\em third-party} cookies that advertising technology platforms like DSPs use to perform targeting, optimization, and measurement of digital advertising, These are distinct from {\em first-party} cookies associated with the publishers of the websites in question. Though first-party cookies tend to be more robust, they are also subject to deletion, and in any case it is typically third-party cookies that power measurement solutions, and hence we refer to these simply as "cookies."}.

Another common type of userID is a mobile advertising ID (such as an Apple IDFA or Google Android ID), which is used to identify the mobile device (commonly, a smart-phone) used by a consumer. These differ from cookies in two principal ways: (1) unlike cookies, mobile advertising IDs are stable, i.e., persistent over time\footnote{A userID can be considered stable for our purposes if the time-frame over which they persist is significantly greater than the PV window, $V_c$, of the campaign in question.}, and (2) mobile advertising IDs identify a particular device, whereas cookies identify a particular instance of a browser on a particular device (and consumers may in fact use more than one browser on a given device; this is actually a common behavior). Consumers are increasingly using multiple connected devices -- from phones, to tablets, to watches -- with that trend expected to continue. Thus, while a device has a unique and persistent mobile advertising ID, a consumer may be associated with several mobile advertising IDs. 
Bid opportunities on different mobile devices owned by the same consumer can therefore have different (mobile advertising ID-based) userIDs, and once again this leads to a situation where some userIDs for a given consumer are assigned to Test and others to Control. We refer to this as \textit{cross-device contamination}.

Both cookie contamination and cross-device contamination are highly prevalent, and are not mutually exclusive with respect to a given consumer. Though the underlying mechanisms differ between them -- cookie contamination arising from the inability of cookie-based userIDs to persistently identify a particular browser across bid opportunities and conversions, and cross-device contamination arising from multi-device consumers appearing with different mobile advertising ID-based userIDs across bid opportunities and conversions -- both effectively result in a situation whereby a given consumer has many different userIDs.\footnote{
The reverse scenario is also possible, where a given userID (either cookie or mobile advertising ID) may correspond to multiple consumers (as might be the case with shared household usage of a single device), though the single-identifier-to-many-consumers scenario is in practice relatively minor compared to the many-identifiers-to-single-consumer scenario.} In particular, if a Test-group userID for a given consumer is exposed to an ad campaign, other userIDs associated with that consumer could subsequently exhibit the response rates of an exposed userID, including ones in Control. In fact, this is another form of SUTVA violation (as described in Section \ref{sec-repeat}): there is interference among units, and specifically, the exposure status of a userID influences the potential outcomes of another userID associated with the same actual consumer. 

Intuitively, one expects that either form of ID contamination would lead to a \textit{dilution} in the measured lift, i.e., the measured userID-level lift would \textit{understate} the true consumer-level lift. The following subsection uses a simple model to quantitatively illustrate this effect, and will further show that the impact can go beyond mere dilution.

\subsection{Impact of Cross-device Contamination}
Several authors have explored the quantitative impact of ID contamination on advertising effectiveness. For example, \cite{coey2016people} present an analysis of the impact of cookie contamination on advertising effectiveness measurement, using both simulations and analytical methods. They demonstrate that cookie-based measurement systematically underestimates the true underlying "person-level" treatment effect, and further show that, under reasonable assumptions, the attenuation factor equals the average number of cookies per user. Based on data from Facebook, the authors of that paper conservatively estimate the average number of cookies per user, over a one-month period, to be at least 3. MediaMath's own data suggests this number can easily be as high as 6 or 7 over a one-year time-frame, with a long tail of users having significantly many more. Thus, cookie contamination clearly has the potential to significantly reduce the signal-to-noise ratio in Ad Lift measurements.

For the purposes of illustration, we consider a simple "toy model" scenario (illustrated in Figure \ref{fig-contam-dilution}) describing cross-device contamination and its directional effects on Ad Lift measurement:
\begin{enumerate}
	\item Each consumer has exactly $k$ device-based userIDs (or just "devices" for short). We assume that on a given device, the userID remains fixed during the period of analysis. We can think of devices as partitioned into clusters of $k$ devices each, and each cluster corresponds to an actual consumer.
	\item Each device is randomly assigned to Control ($C$) with probability $p$, and to Test ($T$) with probability $1-p$.
	\item Win rate is 100\%, i.e., \textit{all} Test devices are exposed to the ad campaign under study.
	\item There are two types of consumers: type-$C$, i.e., those with all $k$ devices in Control (this happens with probability $p^k$), and type-$T$, i.e., those with {\em at least} one device in Test (which happens with probability $1-p^k$). Since every Test device is exposed to an ad, the \textit{userID}-level randomized assignment to Test/Control results in a random assignment to type-$T$ and type-$C$ at the \textit{consumer} level, with the probability of assignment to type-$C$ being $p^k$.
	\item A type-$T$ consumer has a conversion probability of $r_t$, and is \textit{equally likely} to convert from any of their $k$ devices\footnote{Note that we are implicitly assuming here that a consumer's response probability is flat regardless of how many times they are exposed to an ad campaign. In reality, we expect response to vary as a function of frequency of exposure, but we assume frequency independence here in order to arrive at a simple, quantitative, directional estimate of the impact of cross-device contamination.}, i.e., the conversion probability for each device is $r_t/k$. 
	\item A type-$C$ consumer has a conversion probability of $r_c$, and is \textit{equally likely} to exhibit a conversion on \textit{each} of their $k$ devices, i.e. each device has a conversion probability of $r_c/k$.
	\item The true consumer-level lift (in the sense of $ATL$, Definition \ref{def-lift}) is $a$, which means $r_t = r_c(1+a)$.
\end{enumerate}

\begin{figure}\centering
\includegraphics[width=0.9\textwidth]{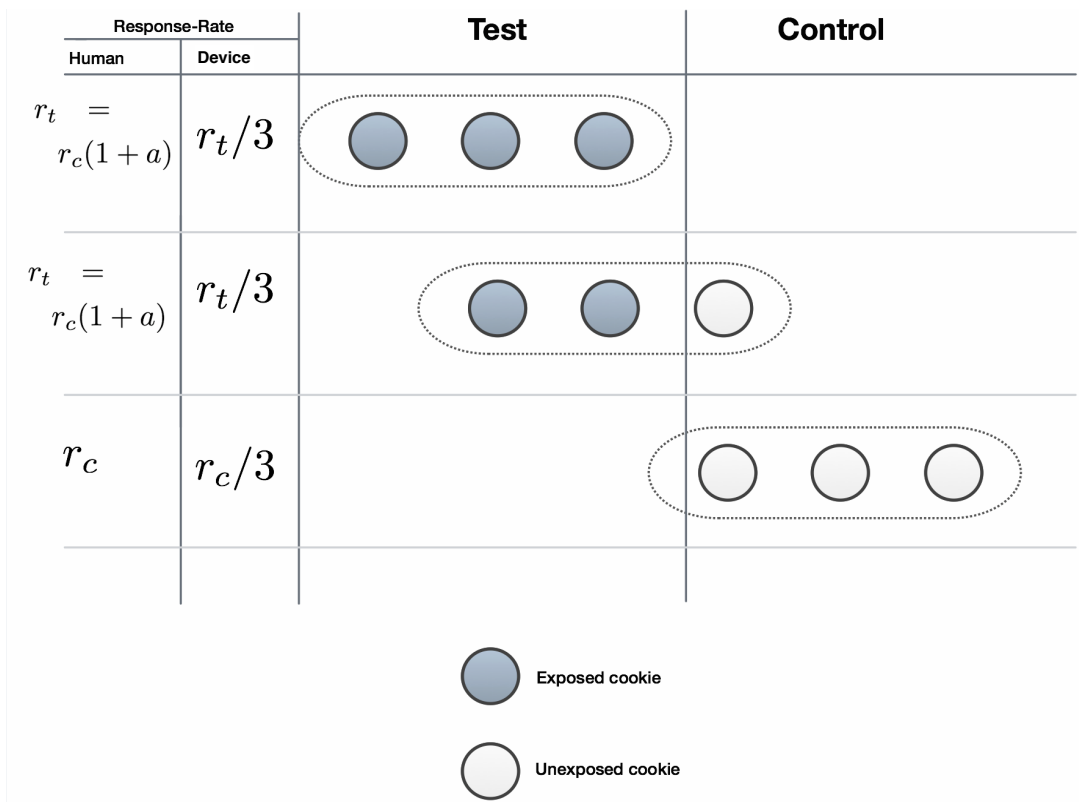}
\caption{\small A simple cross-device contamination scenario where each consumer has exactly three devices, and each device (represented by a circle) is randomly assigned to Control with probability $p$ (and to Test with probability $1-p$). The device distributions of three consumers are shown. Each row of devices is associated with a distinct consumer. The two Response Rate columns show the response rates of the consumer, and that of each device of the consumer.}
\label{fig-contam-dilution}
\end{figure}

We can then show the following:

\begin{proposition}\small [Dilution due to Cross-device Contamination]
	In the above scenario where each consumer has $k$ devices, and each device is assigned to Control with a probability $p$, all Test devices are exposed to ads, and the true consumer-level causal Ad Lift (in the sense of $ATL$, Definition \ref{def-lift}) is $a$, then the lift measured at device (i.e., userID) level will be 
\begin{equation} \label{eq-atl-contam}
	ATL = 	a p^{k-1}/(1 + a - a p^{k-1})
\end{equation}	

\end{proposition}
\begin{proof}\small
Let $R_t$, $R_c$ be the measured Test and Control response rates respectively, at the device-level.
All Test devices must belong to a type $T$ consumer and hence have conversion probability $R_t = r_t/k$. 
On the other hand, the conversion probability of a Control device will depend on whether or not it is "contaminated", i.e., whether or not a "sibling" of that device (i.e., one that belongs to the same consumer) is in Test.
A given device for a Control consumer has probability $1-p^{k-1}$ of being contaminated (in which case its response rate would be $r_t/k$), and probability  $p^{k-1}$ of being "clean" (i.e., all its siblings are in Control, in which case its response rate would be $r_c/k$). 
Thus the measured Control device response rate would be 
\begin{eqnarray}
R_c & = & p^{k-1} r_c/k + (1- p^{k-1})r_t/k \\
    & = & p^{k-1} r_c/k + (1- p^{k-1})r_c(1+a)/k \\
    & = & (1 + a - a p^{k-1})r_c/k,
\end{eqnarray}
which implies 
\begin{eqnarray}
R_t - R_c & = & r_c a p^{k-1}/k,
\end{eqnarray}
and so the lift is given by 
\begin{eqnarray}
ATL & = & (R_t - R_c)/R_c \\
        & = & a p^{k-1}/(1 + a - a p^{k-1}).
\end{eqnarray}
\end{proof}

The above result shows that as $k$ (the number of devices per consumer) increases, the measured device-level lift drops \textit{exponentially} in $k$, since the term involving $p^{k-1}$ in the denominator can be neglected for large $k$ values, and the expression reduces to $a p^{k-1}/(1 + a)$. As a sanity check we can verify that the $ATL$ expression (\ref{eq-atl-contam}) reduces to $a$ when $k=1$. It is also worth noting that the above toy model assumed (as stated in assumption 1 above) that while consumers have $k$ devices each, the identifiers on those devices are stable (i.e., persist) for at least the time between the bid opportunity and the desired outcome. In fact, this is often not the case, and if some conversions occur on devices {\em after} the userIDs for those device have changed (e.g., as could happen in the case of cookie contamination), then these conversions are effectively "lost", in the sense that they would not be attributed to a (Test or Control) bid opportunity. This would result in further depression of both the Test and Control response rates (assuming that cookie non-acceptance and cookie deletion are equally common among Test and Control consumers), effectively decreasing the signal-to-noise ratio even further.

The above illustration might lead one to conclude that at worst, ID contamination would drive the Causal Lift estimation to zero. This makes intuitive sense: when contamination is severe, every Control device essentially has the response rate of a Test (exposed) device, and \textit{as long as the composition of the type-$C$ and type-$T$ consumer populations is the same}, the measured causal Ad Lift would be at worst zero. In fact, ID contamination can not only depress Causal Lift to zero; under realistic conditions it can transform positive actual Ad Lift at the consumer level into negative measured Ad Lift at the userID level, as we will now describe.

In the above scenario, zero (or extremely low) measured lift resulted from high $k$ under the assumption that each consumer has {\em exactly} $k$ devices. If we instead make the more realistic assumption that the population consists of some consumers with just one device, as well as others with multiple devices, then we can show, in general, that consumers with multiple devices will constitute a higher fraction of the Control population than of the Test population. The following result shows a simple example where a non-constant device distribution can result in multi-device human users being over-represented in Control compared to Test.
\begin{proposition}\small [Multi-Device Skew] \label{prop-md-skew}
Consider a scenario where each device is randomly assigned to Control with probability $p$ (and to test with probability $1-p$). For simplicity, suppose there are two kinds of consumers: "1D consumers", i.e., those with just one device, and "2D consumers", i.e., those with exactly two devices. Let the random-variable $T_i$ represent the number of "$i$-D" consumers with at least one device in Test, and $C_i$ represent the number of "$i$-D" consumers with at least one device in Control. Then for a sufficiently large population size the expected value of $C_2/C_1$ is $(2-p)/(1+p)$ times larger than the expected value of $T_2/T_1$, i.e., the ratio of 2D to 1D users in Control is $(2-p)/(1+p)$ times larger than the corresponding ratio in Test.
\end{proposition}
\begin{proof}\small
Suppose there are $N_1$ 1D consumers, and $N_2$ 2D consumers. For each 1D consumer, the chance that their device is in Control is $p$. For each 2D consumer, the probability that \textit{at least one} of their devices is in Control is $1 - (1-p)^2$. Therefore, 
\begin{eqnarray*}
	E[C_1] &=& N_1 p\\
	E[C_2] &=& N_2 (1 - (1-p)^2) = 2p - p^2,
\end{eqnarray*}
It is possible to show using a Taylor expansion (see, e.g., \cite{taylor-stack-exchange})
that for sufficiently large $N_1, N_2$, the expectation of the ratio $C_2/C_1$ can be closely approximated by the ratio of the expectations, i.e., 
\begin{eqnarray*}
E[C_2/C_1] \simeq E[C_2]/E[C_1] = (N_2/N_1)(2p-p^2)/p = (N_2/N_1)(2-p).	
\end{eqnarray*}

Similarly, for each 1D consumer, the chance that their device is in Test is $1-p$, and for each 2D consumer, the chance that at least one of their devices is in Test is $1 - p^2$. This implies that
\begin{eqnarray*}
	E[T_1] &=& N_1 (1-p)\\
	E[T_2] &=& N_2 (1 - p^2),
\end{eqnarray*}
and once again we use the approximation
\begin{eqnarray*}
E[T_2/T_1] \simeq E[T_2]/E[T_1] = (N_2/N_1)(1-p^2)/(1-p) = (N_2/N_1)(1+p),
\end{eqnarray*}
from which the result follows:

\begin{eqnarray*}
E[C_2/C_1] \simeq E[T_2/T_1] (2-p)/(1+p)
\end{eqnarray*}

\end{proof}

In practical randomized measurements of Causal Lift, the probability $p$ of a device (or cookie, for that matter) being assigned to Control is much smaller than 1, and $p = 10\%$ is typical. For $p = 0.10$, the above result shows that the 2D to 1D ratio in Control is $(2-p)/(1+p) = 1.9/1.1 = 1.7$ times the corresponding ratio in Test, meaning that 2D human users are over-represented in control, relative to test, by a factor of 1.7. 

The phenomenon of consumers with multiple devices constituting a higher fraction of the Control population than Test is merely a mathematical consequence of a non-constant device distribution. However, we have also observed in our data that multi-device consumers  have a higher baseline (i.e., no-exposure) response rate (on \textit{each} of their devices) than single-device consumers. This may be because consumers with multiple devices are more "prolific" internet users and/or perhaps more affluent, and as such tend to perform conversion actions at a higher rate than single-device consumers, even in the absence of any advertising exposure. Whatever the reason behind this empirical observation, if we simply assume for the sake of argument that multi-device consumers have a higher baseline response rate, and combine that with the mathematical result that these users are over-represented in Control, it is clear that the Control-group response rate would be \textit{higher} than that of the Test group, thus leading to a \textit{negative lift} measurement. We previously saw that negative lift could result from win bias and corrected for that via our methodology for Pre-Bid Randomization. Here we see that ID contamination can not only drown a positive signal to zero, but also potentially flip a positive signal to a negative one. Clearly Ad Lift measurement solution can be valid without addressing ID contamination.

\subsection{Remedies for Contamination}
In light of the above results showing how ID contamination leads to a degradation of lift measurement,  it is important to consider possible remedies that can account for ID contamination. We now consider how a simple modification of our methodology offers one such remedy. Perhaps unsurprisingly, the modification relies on the availability of information linking together identifers (i.e., userIDs) that belong to the same underlying consumer. At MediaMath this information is integrated into a solution which we refer to as \textit{Connected-ID}, or CID for short. The reader can think of a CID as the "master" ID connecting all the various userIDs (cookies, mobile advertising IDs, etc.) associated with a particular consumer. The primary source of information enabling these connections is login IDs, typically corresponding to email addresses that consumers use to log into accounts on advertiser websites or mobile apps, with the idea being that these login IDs (suitably hashed) are both unique (by defintion) and peristent over time\footnote{Though email addresses may of course change over time, it is assumed that the timescale over which ad effectiveness is measured (i.e., the PV window, $V_c$) is much shorter than the timescale over which consumers typically change email addresses, and thus that emails addresses are effectively stable for these purposes.} and thus when captured simultaenously with various userIDs, provide a way to associate those userIDs to the same underlying consumer. Note that the CID solution is \textit{deterministic}, in contrast to less accurate \textit{probabilistic} approaches which model device data, consumer behaviors, and other information to determine a probability that two or more userIDs are linked to the same underlying consumer. While it is beyond the  scope of our paper to produce a full account of the performance of probabilisitc solutions, our own analysis and the experience of our clients is that the precision and recall of these solutions, as measured against truth sets (where the underlying identity of each consumer is known), is such that they are not reliable enough for accurate targeting of consumers nor for accurate measurement of Ad Lift. Henceforth, we shall use the term CID to refer not only to MediaMath's solution, but to any solution that connects userIDs to a master ID corresponding to the underlying consumer, on the basis of deterministic data.

In general, no CID solution will be complete, meaning that while it may capture most or all of the userIDs for some consumers, it will not generally be able to resolve all userIDs to a CID for all consumers. In other words, some userIDs will have no CID mapping, e.g., because login ID information has not been captured for that consumer and joined to the CID graph, or because the device is new, etc. Despite incomplete (and possibly noisy) CID information, we can use a {\em sufficiently large} CID solution to improve our lift estimates with a simple modification: 

\begin{modification}[Treat the CID as a unit]
Discard userIDs for which there is no CID mapping. For the remaining userIDs, use the CID as the identifer instead of the userID. In particular, in the Pre-Bid Test/Control assignment, use the hash of the CID rather than the original userID. For all the Causal Lift calculations, ignore the original userID and use the CID instead. This will then yield a Causal Lift measurement at the CID level (or consumer level).
\end{modification}

If we assume that the CID information is \textit{complete} (i.e., it captures \textit{all} userID linkages), and \textit{correct} (i.e., all linked userIDs in fact belong to the same consumer), then it should be clear that there is no SUTVA violation (i.e., no interference among units, which are CIDs), and hence the methodology in Section \ref{sec-causal-est} works. 

Of course, in reality the CID information may be incomplete and incorrect. Incompleteness, however, is only an issue to the extent that it drives uncertainty in the estimation of Ad Lift and the confidence intervals around it, using the methodologies presented earlier. The more complete the CID solution is (i.e., the more userIDs it can link to CIDs), the smaller those errors will be. The CID solution merely needs to be {\em sufficiently large} enough to produce measurement results with sufficiently high confidence. How large does this need to be in practice? MediaMath's CID solution currently has 63 million CIDs in the US, corresponding to consumers with more than one userID. This is about 25\% of the estimated 250 million unique consumers in the US. As we will demonstrate in the next Section, this data has yielded Ad Lift measurement with high confidence. While this does not provide a lower bound on the absolute or relative size of the CID solution needed, it does demonstrate that a complete solution is not needed. For a given CID solution, one can simply evaluate the resulting Ad Lift and confidence intervals to empirically assess whether it is "large enough."

As for correctness, most CID solutions employ a variety of techniques to ensure the quality of the data is high. These often include supervised testing against known deterministic IDs, unsupervised testing to flag and filter anomalies (e.g., abnormally high numbers of userIDs mapping to a single CID), performance validation (i.e., looking at actual business results derived from the application of CID solutions to marketing activities), and other methods. A complete framework for evaluating the correctness of an arbitrary CID solution is beyond the scope of this paper. However, it is generally possible to increase the correctness of the CID information by sacrificing some coverage with the following further heuristic: 

\begin{modification}[Restrict the analysis CIDs with $k$ or more userIDs]
In addition to restricting to userIDs that have a (deterministic) mapping to a CID, further restrict to CIDs that link $k$ or more userIDs. 
\end{modification}

The practical justification for this modification is twofold: 
\begin{enumerate}
\item It removes singleton instances where a CID is mapped to only one userID. While this may correctly indicate that this consumer has only one device, we have noted in practice that such "one-hit wonder" occurrences often tend to arise from a variety of issues associated with failure to establish correct mappings. 
\item The larger the value of $k$, the more likely it is that each CID has connected most or all of the userIDs corresponding to a consumer, i.e., the more confident one can be in the CID solution's ability to correctly associate userIDs to CIDs.
\end{enumerate}

Clearly there is a trade-off with this approach: for higher $k$ we will have less data, resulting in wider confidence bounds around the measured lift. In practice, we have found the best results for $k=2$ or $k=3$, which produce noticeably different results vs. $k=1$. This is consistent with empirical studies (e.g., \cite{coey2016people}) showing that a vast majority ($> 70\%$) of consumers have fewer than 4 cookies, and that digital consumers own, on average, about 3.5 connected devices \cite{emarketer-2016}
though most of their media consumption (and hence ad exposure) occurs on 2 or 3 devices (e.g., smart-phones, laptops, and tablets, as opposed to gaming consoles, wearables, etc.). Thus, using $k=2$ or $k=3$ helps prevent ID contamination from overwhelming lift measurements, avoiding singletons and gearing the analysis around the empirically expected value for $k$.

We believe that the issues discussed in this paper -- notably around ID contamination and win bias -- have severely hampered previous attempts to measure Ad Lift. In MediaMath's experience, having observed many dozens if not hundreds of attempted measurements, spanning different advertiser verticals and different campaigns configurations, the results routinely come back indicating marginal lift, no lift, or negative lift. These results (in particular negative lift), run counter to common-sense expectations about the impact of advertising, and fly in the face of many billions of dollars of ad spend. The few tests that we have seen produce strong, positive lift tended to do so only for short-periods of time, quickly falling back into the noise. We believe this is because those approaches, unlike the current work, do not account for real-world constraints and complications. In the next Section, we present Ad Lift measurements using the methodology presented here, which yields significant, positive, and stable measurements.

\section{Experimental Results} \label{sec-results}

In table \ref{tab-results} we show results of the Causal lift ($ATL$) computation using the modified methodology of Section \ref{sec-contam}, for seven campaigns using bid opportunity data over a 30-day period ending March 21, 2017. 
The table shows various measures related to $ATL$ as well as the raw counts of Test and Control unique CIDs and conversions, and several intermediate values that illustrate our methodology. 
We now describe the columns in the table in some detail:

\begin{itemize}
	\item \textbf{id} is the ID of an ad campaign
	\item $ATT$ is the Average Treatment Effect on Treated, or $R_{TW} - R_{CW}$ as defined in Eq. \ref{eq-att-noncomp}. Note that by Lemma \ref{lem-att-prebid} $ATT = (R_T - R_C)/w$, where $w$ is the win rate (see below).
	\item $ATL$ is the Causal Lift as defined in Eq. \ref{eq-atl}, i.e., 
\begin{equation*}
	ATL = (R_{TW} - R'_{CW})/R'_{CW} = ATT/(R_{TW} - ATT) 
\end{equation*}
	\item $INC$ is the \textit{Incrementality}, as defined in Eq. \ref{eq-inc}, i.e., 
\begin{equation*}
	INC = (R_{TW} - R'_{CW})/R_{TW} = ATT/R_{TW}.
\end{equation*}

	\item The column names starting with "g" denote quantities related to the Gibbs-sampling methodology in Section \ref{sec-gibbs}: $g5$, $g50$, $g95$ denote the 5th percentile, 50th percentile (or median), and 95th percentile, respectively, of the distribution of possible values of $ATL$ that are consistent with the observed data, as generated by the Gibbs Sampling procedure. So one can treat the range $(g5, g95)$ as the 90\% confidence-interval for $ATL$. $gConf$ is the Gibbs-sampling-based "directional confidence" in the observed lift ($ATL$): we model the distribution of $ATL$ as a Gaussian with mean $\mu$ and standard-deviation $\sigma$ computed from the generated values of $ATL$, and then define $gConf$ as one minus the probability that we would have observed an $ATL$ as high as the one we observed, had the true value of $ATL$ been negative, under the assumed Gaussian distribution. 
	\item The 6 columns $TU, TC, TWU, TWC, CU, CC$ represent the raw observed counts that are the inputs to our methodology. $TU, TC$ are respecitvely the number of unique CIDs and conversions in the Test group. $TWU, TWC$ are the number of unique CIDs and conversions in the $TW$ (Test-Winners) population. $CU, CC$ are respecitvely the number of unique CIDs and conversions in the Control population. 
	\item $w$ is the win rate, defined as $TWU/TU$, expressed as a percentage.
	\item $R_T, R_C$ are the observed response rates of the Test and Control populations, respectively, expressed as percentages: $R_T = TC/TU$, $R_C = CC/CU$.
	\item $R_{TW}$ is the response rate of the Test-Winners group, equal to $TWC/TWU$, expressed as a percentage.
\end{itemize}

All of the various response rates, and $ATT$ (which is a difference in response rates $R_{TW} - R_{CW}$) are shown in the table in \textit{basis points} (bp), where one bp is 0.01\%.

\begin{sidewaystable}[!h] 
\small
\caption{\small Causal Lift ($ATL$) measurement results, along with various related computations and intermediate values, for seven campaigns over a 30-day period ending on March 21, 2017.}
\label{tab-results}
\begin{tabular}{|l|l|l|l|l|l|l|l|l|l|l|l|l|l|l|l|l|l|}
\hline
id & ATL & INC & ATT & gConf & g5  & g50 & g95  & $R_T$   & $R_C$   & $R_{TW}$  & $w$  & TU      & TC     & TWU    & TWC   & CU     & CC    \\
\hline \hline 
1   & 63  & 39  & 25  & 99    & 12  & 62  & 167  & 63   & 49   & 65   & 56 & 263,501  & 1,670   & 148,058 & 955   & 16,065  & 79    \\
2   & 18  & 15  & 14  & 100   & 7   & 18  & 30   & 57   & 52   & 93   & 42 & 2,195,456 & 12,609  & 918,316 & 8,573  & 145,216 & 748   \\
3   & 17  & 14  & 8   & 91    & -3  & 16  & 44   & 46   & 43   & 53   & 47 & 734,135  & 3,390   & 346,656 & 1,840  & 69,511  & 296   \\
4   & 534 & 84  & 6   & 100   & 74  & 335 & 3597 & 3    & 2    & 8    & 13 & 4,938,065 & 1,423   & 657,002 & 503   & 459,553 & 93    \\
5   & 153 & 60  & 6   & 94    & -11 & 109 & 1317 & 4    & 3    & 9    & 15 & 2,409,520 & 902    & 364,234 & 343   & 110,991 & 32    \\
6   & 5   & 5   & 46  & 100   & 3   & 5   & 8    & 629  & 610  & 887  & 40 & 1,955,475 & 122,968 & 787,613 & 69,874 & 205,131 & 12,520 \\
7   & 12  & 11  & 1   & 65    & -23 & 9   & 72   & 2    & 2    & 6    & 24 & 2,833,414 & 511    & 681,506 & 380   & 198,932 & 33    \\
\hline
\end{tabular}
\end{sidewaystable}

To illustrate how our methodology is applied, let us consider the calculations for campaign 1 (top row in the table):

\begin{example}\small [Calculation of $ATL$ for campaign 2]
We start with the six raw observed counts with values of $TU=263,501, TC=1,670, TWU=148,058, TWC=955, CU=16,065, and CC=79$. First we compute these response rates:
\begin{eqnarray*}
	R_{T} &= & TC/TU = 1670/263501 = 0.0063\% = 63bp\\
	R_C   &= & CC/CU = 70/16065 = 0.0049\%  = 49bp\\
	R_{TW} &=& TWC/TWU = 955/148058 = 0.0065\% = 65bp, \\
\end{eqnarray*}
and the win rate $w = TWU/TU = 56\%$.
We then compute the $ACE$, or Average Causal Effect (also called the Intent-to-Treat effect, $ITT$), as in Eq. \ref{eq-ace-noncomp}:
\begin{equation}
	ACE = R_T - R_C = 63bp - 49bp = 14bp,
\end{equation}
which means that \textit{at the Intent-to-Treat level}, the effect of the ad campaign is 14bp. From Eq. \ref{eq-att-noncomp} we see that the $ATT$ (Average Treatment Effect on Treated), is 
\begin{equation}
	ATT = ACE/w = 14bp/0.56 = 25bp = (R_{TW} - R_{CW}).
\end{equation}
This shows that the casual effect of the ad campaign \textit{on the exposed population}	is 25bp. Since $ATT$ is (by definition) equal to $R_{TW} - R_{CW}$, we infer that $R_{CW} = R_{TW} - ATT = 65bp - 25bp = 40bp$.
Now we can calculate $ATL$ and $INC$ (Eqs \ref{eq-atl} and \ref{eq-inc}) as follows:
\begin{eqnarray*}
	ATL &=& ATT/R_{CW} = 25/40 = 63\%	\\
	INC &=& ATT/R_{TW} = 25/65 = 39\%,
\end{eqnarray*}
which shows that the response rate of the exposed Test population is 63\% higher than that of the counterfactual Control winner population, and that out of the 65bp response rate of the exposed Test population, only 39\% is causally attributable to the ad campaign.
\qedsymbol
\end{example}

Looking across the campaigns, we can see a wide range of $ATL$ and $INC$ values, presumably owing to the different intrinsic effectiveness of these campaigns against their target consumers. Of note, the values for $ATL$ and $INC$ are all positive and the values of $gConf$ are above $90\%$ for all but one campaign, illustrating the utility of providing not only point estimates but also confidence measurements. We reiterate that these results are in stark contrast to our experience with other methodologies, which tend to produce marginal lift, no lift, or negative lift, and often with a high degree of inconsistency.

\section{Related Work}\label{sec-related}

The \textit{ghost ads} methodology \cite{Johnson_Ghost_2015} aims to estimate the causal/incremental effect of ads by \textit{identifying the counterfactual} winner types in the Control population. This hinges on being able to accurately simulate an auction and predict whether or not a Control consumer \textit{would have} seen the ad (i.e., would have \textit{won} the auction if the consumer had been assigned to Test).
This may be possible for a small handful of so-called “walled garden” media companies such as Google, Facebook, and Amazon, who can see  both the “buy” and “sell” sides of the auction process and dictate the outcome according to their internal algorithms and specifications, but in general for ad buyers this is not possible and such an approach cannot be implemented.
By contrast our methodology uses randomization and a correction for auction win bias to avoid the need to actually identify counterfactual winner types in the Control group. 

There is a strand of literature addressing the challenging problem of estimating the causal effect of advertising based on \textit{observational} data, i.e., from data collected during the normal course of campaign delivery, without setting up explicit randomized control tests. One such paper is \cite{stitelman2011estimating} which considers several ways to address the chief difficulty in estimating causality from observational data: the exposed and unexposed populations will not in general be equivalent. The authors consider various ways of correcting for this non-equivalence. Essentially these approaches rely on doing a counterfactual or "what-if" analysis \textit{analytically} via fitted models/estimators rather than experimentally via randomization. An important requirement in these approaches is that \textit{all} potentially confounding variables (i.e., those that affect assignment to the treatment group and also affect response rates) are observable, which is rather a strong requirement and one which we believe does not hold in the real world. Another prominent paper that discusses the observational approach to causality in the digital advertising context, is \cite{chan2010evaluating}. An interesting survey paper from Facebook \cite{gordon2016comparison} shows that observational methods often fail to produce results as good as true randomized experiments, even after taking into account thousands of confounding variables. In fact one of the methods presented in this paper is similar to our method for computing the $ATT$ (Eq \ref{eq-att-noncomp}). 

We believe observational  approaches are best-suited to estimating impact at the individual/unit level (i.e., $ICE$, Definition \ref{ice}), where counterfactuals cannot be established since individual consumers can of course only receive either the Test or Control treatment, but not both. Indeed, such individual-level estimates are needed, for example, when trying to determine the correct price to bid in a real-time for a bid opportunity, based on the expected incremental impact of that particular ad exposure against that particular consumer. At the aggregate level, however, observational approaches have several drawbacks: they are generally assumption-driven (in terms of the models/estimators used), they are prone to unknown biases in the construction of treatment populations, and they have proven to be difficult to validate in practice, leaving many advertisers questioning if their results can really be believed. While they do have the appeal of zero opportunity cost to advertisers (who would miss out showing ads to the Control group in a randomized test), we believe these downsides make such approaches sub-optimal for reliable Incrementality measurement at an aggregate level. By contrast, an experimental approach at the aggregate level (where randomization is well defined, and statistically equivalent counterfactual Test and Control groups can be established) -- provided it can account for the real-world factors addressed in this work -- can produce assumption-free, unbiased estimates that can be easily reproduced and verified.

 We note that in the category of observational approaches, there has been recent interest in applying Machine Learning techniques to estimate the \textit{Individual Causal Effect} $ICE$ (or Individual Treatment Effect, $ITE$) as a function of covariates (or features) of an individual entity. For example, \cite{athey2016recursive, Athey2015} present a decision-tree based approach to estimate the $ICE$; \cite{johansson2016learning} show how counterfactuals can be estimated by learning representations (using linear models or deep neural networks) that encourage similarity (or balance) between Test and Control populations; \cite{hartford2016counterfactual} introduce the concept of "Deep Instrumental Variables Networks" which are able to predict the counterfactual and hence estimate causal effects.

Two papers present the idea of using Gibbs Sampling in the context of causal analysis: \cite{Chickering1996} outline a method to apply Gibbs sampling to causal inference in the presence of non-compliance, and \cite{Barajas2012} present a similar method, applied to advertising campaign effectiveness measurement. Neither of these papers give a complete specification that can be easily implemented. Our Gibbs Sampling procedure is similar to these but we have introduced important simplifications, and specified the approach in a manner that is self-contained and easy to implement.

\section{Future Work}\label{sec-future}

In future work, we intend to explore various generalizations and extensions of this approach. These include: 1) adapting the methodology to estimate the Ad Lift from multiple campaigns and/or digital advertising channels, taking into account interaction effects, 2) specifying how such Ad Lift measurements should correctly inform advertising budget allocation decisions, and 3) Unifying this approach with a methodology for the valuation of individual bid opportunities, in a manner that accounts for incrementality at the bid-opportunity level.

\section{Acknowledgements}\label{sec-acknowledge}

We would like to express our gratitude to Himanish Kushary, Jason Lei, Jonathan Marshall,  and Caryl Dizon Yuhas for numerous helpful discussions and problem-solving sessions.

\bibliography{papers-mar29,all_papers,mendeley-bib}{}

\bibliographystyle{apalike}

\end{document}